\documentclass[11pt,twoside]{article}
\usepackage{amsmath, amssymb, amsthm, amsfonts,cite,alltt,clrscode}
\usepackage[english]{babel}
\usepackage{graphicx}
\usepackage{caption}
\usepackage{url}
\usepackage{hyperref}
\usepackage[list=true,listformat=simple]{subcaption}
\usepackage{wrapfig}

\let\epsilon=\varepsilon

\usepackage{cmap}
\usepackage[T1]{fontenc}

\newtheorem{theorem}{Theorem}[section]
\newtheorem{metatheorem}[theorem]{Metatheorem}

\newtheorem{corollary}[theorem]{Corollary}
\theoremstyle{definition}

\newtheorem{problem}[theorem]{Problem}

\newif\ifabstract
\abstractfalse
\newif\iffull
\ifabstract \fullfalse \else \fulltrue \fi

\graphicspath{{./figuresPushPull/},{./figuresPortal/},{./screenshots/},{./images/}}

\usepackage[
	counters = {theorem},
	full
]{magicappendix}
\abstractfalse

\usepackage[margin=1in]{geometry}
\begin{document}

\title{The Computational Complexity of Portal \\ and Other 3D Video Games}
\author{%
  Erik D. Demaine%
    \thanks{MIT Computer Science and Artificial Intelligence Laboratory, 32 Vassar Street, Cambridge, MA 02139, USA, \protect\url{{edemaine,jaysonl}@mit.edu}}
\and
  Joshua Lockhart%
  \thanks{Work started while author was at School of Electronics, Electrical Engineering and Computer Science,
  Queen's University, Belfast, BT7 1NN, UK. 
  Author's current address: Department of Computer Science, University College London, London, WC1E 6BT, UK, \protect\url{joshua.lockhart.14@ucl.ac.uk}}
\and
  Jayson Lynch\footnotemark[1]
}
\date{}
\maketitle

\begin{abstract}
We classify the computational complexity of the popular video games Portal and Portal 2. We isolate individual mechanics of the game and prove NP-hardness, PSPACE-completeness, or (pseudo)polynomiality depending on the specific game mechanics allowed. One of our proofs generalizes to prove NP-hardness of many other video games such as Half-Life 2, Halo, Doom, Elder Scrolls, Fallout, Grand Theft Auto, Left 4 Dead, Mass Effect, Deus Ex, Metal Gear Solid, and Resident Evil.

These results build on the established literature on the complexity of video games \cite{HardGames12, NintendoFun2014, Forisek10, Lemmings04}.
\end{abstract}

%\pagebreak
%
%\tableofcontents 
%\pagebreak
%
%\listoffigures 
%\pagebreak
%
%\listoftables
%\pagebreak

\section{Introduction} 
In Valve's critically acclaimed \emph{Portal} franchise, the player guides \emph{Chell} (the game's silent protagonist) through a ``test facility'' constructed by the mysterious fictional organization Aperture Science.

Its unique game mechanic is the Portal Gun, which enables the player to place a pair of portals on certain surfaces within each test chamber. When the player avatar jumps into one of the portals, they are instantly transported to the other. This mechanic, coupled with the fact that in-game items can be thrown through the portals, has allowed the developers to create a series of unique and challenging puzzles for the player to solve as they guide Chell to freedom. Indeed, the Portal series has proved extremely popular, and is estimated to have sold more than 22 million copies \cite{PortalSales,SteamSpyPortal,Portal2Sales,SteamSpyPortal2}.

We analyze the computational complexity of Portal following the recent surge of interest in complexity analysis of video games and puzzles. Examples of previous work in this area includes the proof of NP-completeness of Minesweeper~\cite{Minesweeper00}, Clickomania~\cite{ClickomaniaGameTheory2000,Clickomania_MOVES2015}, and Tetris~\cite{Tetris03}, as well as PSPACE-completeness of Lemmings~\cite{Lemmings04, viglietta2015lemmings} and Super Mario Bros.~\cite{Mario_FUN2016}.
See also the surveys~\cite{AlgGameTheory_GONC3, NPPuzzles08, GPCBook09}.
Recent work has moved from puzzles to classic arcade games~\cite{HardGames12}, Nintendo games~\cite{NintendoFun2014}, 2D platform games~\cite{Forisek10}, and others~\cite{DBLP:journals/corr/Walsh14, floodIt, DBLP:journals/corr/abs-1203-1633}. 

In this paper, we explore how different game elements contribute to the computational complexity of Portal 1 and Portal 2 (which we collectively refer to as \emph{Portal}), with an emphasis on identifying gadgets and proof techniques that can be used in hardness results for other video games. We show that a generalized version of Portal with Emancipation Grills is weakly NP-hard (Section~\ref{sec:PortalFalling}); Portal with turrets is NP-hard (Section~\ref{sec:PortalTurrets}); Portal with timed door buttons and doors is NP-hard (Section~\ref{sec:PortalTimed}); Portal with High Energy Pellet launchers and catchers is NP-hard (Section~\ref{sec:PortalHEP}); Portal with Cubes, Weighted Buttons, and Doors is PSPACE-complete (Section~\ref{sec:pspace}); and Portal with lasers, laser relays, and moving platforms is PSPACE-complete (Section~\ref{sec:pspace}).

Table~\ref{PortalResultsTable} summarizes these results.
The first column lists the primary game mechanics of Portal we are investigating. The second and third column note whether the long fall or Portal Gun mechanics are needed for the proof. Section~\ref{sec:PortalDefinitions} provides more details about what these models mean.
The turret proof generalizes to many other video games, as described in \iffull Section~\ref{subsec:OtherGames}\else Section~\ref{sec:PortalTurrets}\fi.

\begin{table}[ht]
\centering
\centerline{
\begin{tabular}{|l|l|l|l|}
\hline
\textbf{Mechanics} & \textbf{Portals} & \textbf{Long Fall} & \textbf{Complexity} \\ \hline
\hline
None & No & Yes & P (\S  \ref{sec:PortalNavigation}) \\ \hline
Emancipation Grills, No Terminal Velocity & Yes & Yes & Weakly NP-hard (\S  \ref{sec:PortalFalling}) \\ \hline
Turrets & No & Yes & NP-hard (\S  \ref{sec:PortalTurrets}) \\ \hline
Timed Door Buttons and Doors & No & No & NP-hard (\S  \ref{sec:PortalTimed}) \\ \hline
HEP Launcher and Catcher & Yes & No & NP-hard (\S  \ref{sec:PortalHEP}) \\ \hline
Cubes, Weighted Buttons, Doors & No & No & PSPACE-comp. (\S  \ref{sec:pspace}) \\ \hline
Lasers, Relays, Moving Platforms & Yes & No & PSPACE-comp. (\S  \ref{sec:pspace}) \\ \hline
Gravity Beams, Cubes, Weighted Buttons, Doors & No & No & PSPACE-comp. (\S  \ref{sec:pspace}) \\ \hline
\end{tabular}
}
\caption{Summary of New Portal Complexity Results}
\label{PortalResultsTable}
\end{table}

\ifabstract
\later{
\section{Additional Notes About Portal}
\label{appen:defintions}}

\later{
Viglietta~\cite{HardGames12} argues that most modern games include Turing-complete scripting languages and thus could allow designers to create undecidable puzzles, and Portal (being based on the Source engine) is no exception. However, we feel that many games (including Portal) have consistent game mechanics from which puzzles are built, and that analyzing the complexity of games and puzzles that can be created within those rule sets is interesting and meaningfully captures what we think of as the game they compose.
}
\fi

\section{Definitions of Game Elements}
\label{sec:PortalDefinitions}

Portal is a \emph{platform game}: a single-player game with the goal of navigating the avatar from a start location to an end location of a series of stages, called \emph{levels}. The gameplay in Portal involves walking, turning, jumping, crouching, pressing buttons, picking up objects, and creating portals. The locations and movement of the avatar and all in-game objects are discretized.
For convenience we make a few assumptions about the game engine, which we feel preserve the essential character of the games under consideration, while abstracting away certain irrelevant implementation details in order to make complexity analysis more amenable:
\begin{itemize}
	\item Positions and velocities are represented as fixed-point numbers.\footnote{The actual game uses floats in many instances. We claim that all our proofs work if we round the numbers involved, and only encode the problems in the significand.}
	\item Time is discretized and represented as a fixed-point number.
	\item At each discrete time step, there is only a constant number of possible user inputs: button presses and the cursor position.
	\item The cursor position is represented by two fixed-point numbers.
\end{itemize}

In Portal, a \emph{level} is a description of the polygonal surfaces in 3D defining the geometry of the map, along with a list of game elements with their locations and, if applicable, connections to each other. 
In general, we assume that the level can be specified succinctly as a
collection of polygons whose coordinates may have polynomial precision,
(and thus so can the player coordinates), and thus exponentially large values
(ratios).  This assumption matches the Valve Map Format (VMF) used to
specify levels in Portal, Portal~2, and other Source games \cite{VMF}.
A realistic special case is where we aim for \emph{pseudopolynomial}
algorithms, that is, we assume that the coordinates of the polygons and
player are assumed to have polynomial values/ratios (logarithmic precision),
as when the levels are composed of explicit discrete blocks.
This assumption matches the voxel-based P2C format sometimes used for
community-created Portal~2 levels \cite{P2C}.

In this work, we consider the following decision problem, which asks whether a given level has a path from the given start location the end location.
\begin{problem}
\textsc{Portal}

\textit{Parameter}: A set of allowed gameplay elements.

\textit{Input}:
A description of a Portal level using only allowed gameplay elements, and
spatial coordinates specifying a start and end location.

\textit{Output}: Whether there exists a path traversable by a Portal player from the start location to the end location.
\end{problem}

The key game mechanic, the \emph{Portal Gun}, creates a portal on the closest surface in a direct line from the player's avatar if the surface is of the appropriate type. We call surfaces that admit portals \emph{portalable}. There are a variety of other gameplay elements which can be a part of a Portal level. Because we use many these in our proofs, we describe them in detail below.

\newcounter{papercount}
\begin{center}
    \captionsetup{justification=centering}
    \begin{minipage}{.68\textwidth}
    \centering
\begin{enumerate}
\setcounter{enumi}{\the\value{papercount}}
    \item A \emph{long fall} is a drop in the level terrain that the avatar can jump down from without dying, but cannot jump up. \setcounter{papercount}{\the\value{enumi}}
\end{enumerate}
        \end{minipage}%
        \hfill
    \begin{minipage}{0.28\textwidth}
    \centering
\includegraphics[height=0.11\textheight]{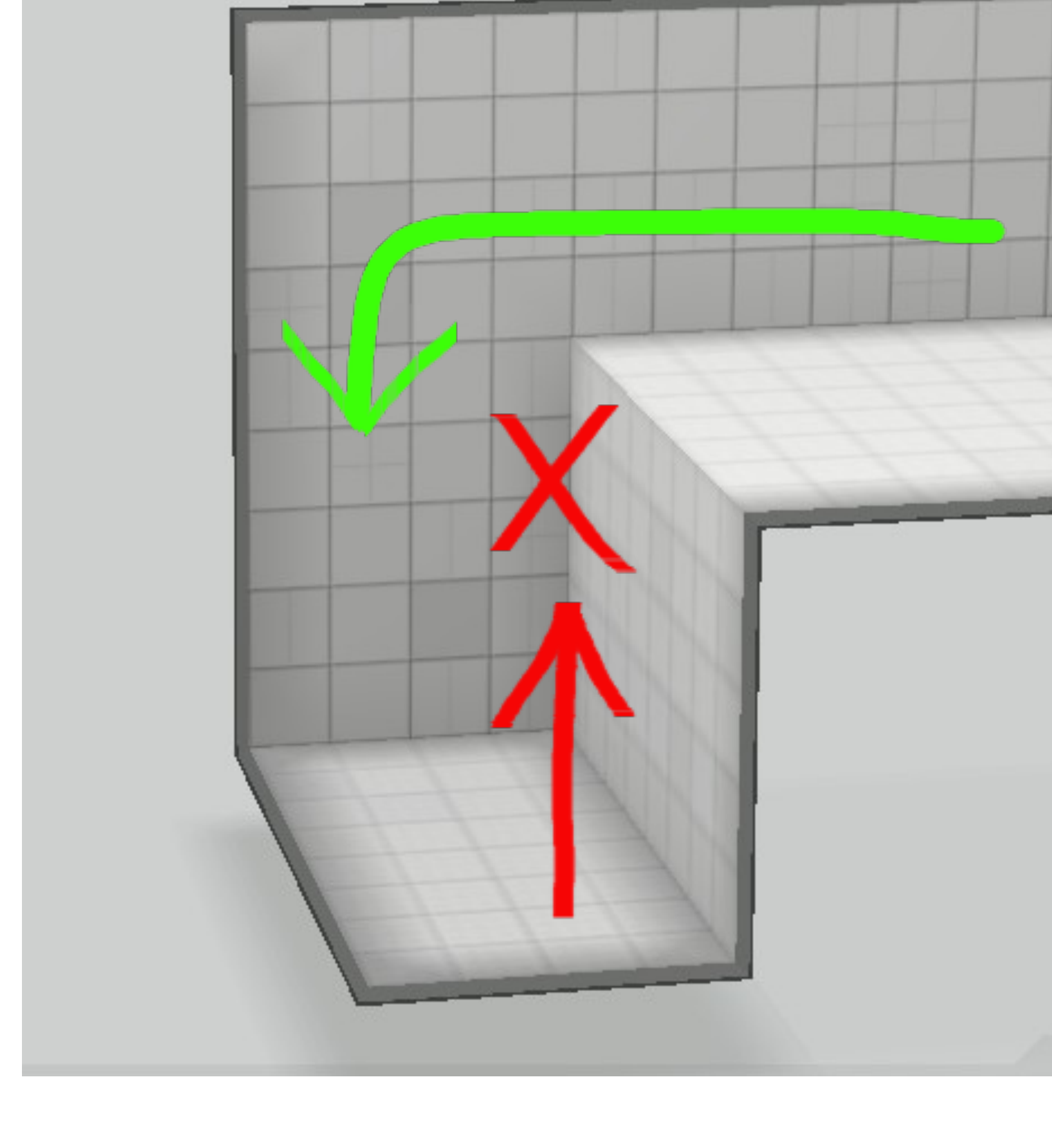}
        \captionof*{figure}{It's a long way down.}
        \label{fig:longfall}
        \end{minipage}
\end{center}

\begin{center}
    \captionsetup{justification=centering}
    \begin{minipage}{.68\textwidth}
    \centering
    
\begin{enumerate}
\setcounter{enumi}{\the\value{papercount}}
    \item A \emph{door} can be open or closed, and can be traversed by the player's avatar if and only if it is open. In Portal, many mechanics can act as doors, such as literal doors, laser fields, and moving platforms. On several occasions we will assume the door being used also blocks other objects in the game, such as High Energy Pellets or lasers, which is not generally true. 
    \setcounter{papercount}{\the\value{enumi}}
\end{enumerate}
        \end{minipage}
        \hfill
    \begin{minipage}{0.28\textwidth}
    \centering
            \includegraphics[width=0.68\linewidth]{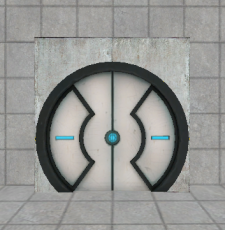}
        \captionof*{figure}{A Door in Portal 2}
        \label{fig:Door}
        \end{minipage}
\end{center}

\begin{center}
    \captionsetup{justification=centering}
    \begin{minipage}{.68\textwidth}
    \centering
\begin{enumerate}
\setcounter{enumi}{\the\value{papercount}}
    \item A \emph{button} is an element which can be interacted with when the avatar is nearby to change the state of the level, e.g., a button to open or close a door.\\ \vspace{2mm}
    \item A \emph{timed button} will revert back to its previous state after a set period of time, reverting its associated change to the level too, e.g., a timed button which opens a door for 10 seconds, before closing it again.
    \setcounter{papercount}{\the\value{enumi}}
\end{enumerate}
        \end{minipage}
        \hfill
    \begin{minipage}{0.28\textwidth}
    \centering
            \includegraphics[width=0.28\linewidth]{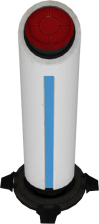}
        \captionof*{figure}{Timed Button}
        \label{fig:Button}
        \end{minipage}
\end{center}

\begin{center}
    \captionsetup{justification=centering}
    \begin{minipage}{.68\textwidth}
    \centering
    \begin{enumerate}
\setcounter{enumi}{\the\value{papercount}}
\item A \emph{weighted floor button} is a an element which changes the state of a level when one or more of a set of objects is placed on it. In Portal, the 1500 Megawatt Aperture Science Heavy Duty Super-Colliding Super Button is an example of a weighted floor button which activates when the avatar or a Weighted Storage Cube is placed on top of it. An activated weighted floor button can activate other mechanics such as doors, moving platforms, laser emitters, and gravitational beam emitters.
\setcounter{papercount}{\the\value{enumi}}
\end{enumerate}
        \end{minipage}
        \hfill
    \begin{minipage}{0.28\textwidth}
    \centering
\includegraphics[width=0.68\linewidth]{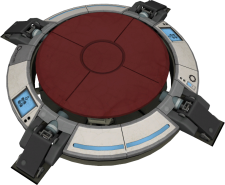}
        \captionof*{figure}{Heavy Duty Super-Colliding Super Button}
        \label{HeavyDutySuper-CollidingSuperButton}
        \end{minipage}
\end{center}

\begin{center}
    \captionsetup{justification=centering}
    \begin{minipage}{.68\textwidth}
    \centering
    \begin{enumerate}
\setcounter{enumi}{\the\value{papercount}}
\item \emph{Blocks} can be picked up and moved by the avatar. The block can be set down and used as a platform, allowing the avatar to reach higher points in the level. While carrying a block, the avatar will not fit through small gaps, rendering some places inaccessible while doing so. In Portal, the Weighted Storage Cube is an example of a block that can be jumped on or used to activate weighted floor buttons. We will refer to Weighted Storage Cubes, Companion Cubes, etc.\ as simply \emph{cubes}.
\setcounter{papercount}{\the\value{enumi}}
\end{enumerate}
        \end{minipage}
        \hfill
    \begin{minipage}{0.28\textwidth}
    \centering
\includegraphics[height=0.11\textheight]{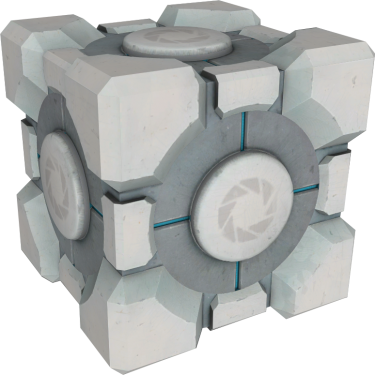}
        \captionof*{figure}{Weighted Storage Cube}
        \label{fig:WeightedStorageCube}
        \end{minipage}
\end{center}

\begin{center}
    \captionsetup{justification=centering}
    \begin{minipage}{.68\textwidth}
    \centering
    \begin{enumerate}
\setcounter{enumi}{\the\value{papercount}}
\item A \emph{Material Emancipation Grid}, also called an \emph{Emancipation Grill} or \emph{fizzler}, destroys some objects which attempt to pass through it, such as cubes and turrets. When the avatar passes through an Emancipation Grid, all previously placed portals are removed from the map.
\setcounter{papercount}{\the\value{enumi}}
\end{enumerate}

        \end{minipage}
        \hfill
    \begin{minipage}{0.28\textwidth}
    \centering
\includegraphics[ height=0.11\textheight]{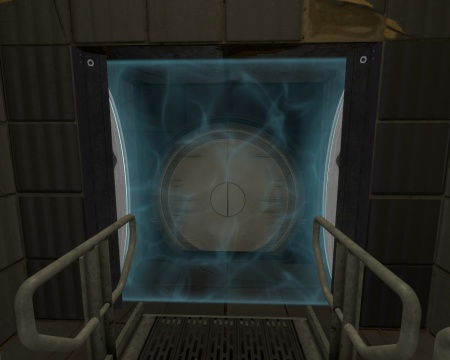}
        \captionof*{figure}{Emancipation Grid}
        \label{fig:EmancipationGrid}

        \end{minipage}
\end{center}

\begin{center}
    \captionsetup{justification=centering}
    \begin{minipage}{.68\textwidth}
    \centering
    \begin{enumerate}
\setcounter{enumi}{\the\value{papercount}}
\item The \emph{Portal Gun} allows the player to place portals on portalable surfaces within their line of effect. Portals are orange or blue. If the player jumps into an orange (blue) portal, they are transported to the blue (orange) portal. Only one orange portal and one blue portal may be placed on the level at any given time. Placing a new orange (blue) portal removes the previously placed orange (blue) portal from the level.
\setcounter{papercount}{\the\value{enumi}}
\end{enumerate}

        \end{minipage}
        \hfill
    \begin{minipage}{0.28\textwidth}
    \centering
            \includegraphics[width=0.5\linewidth]{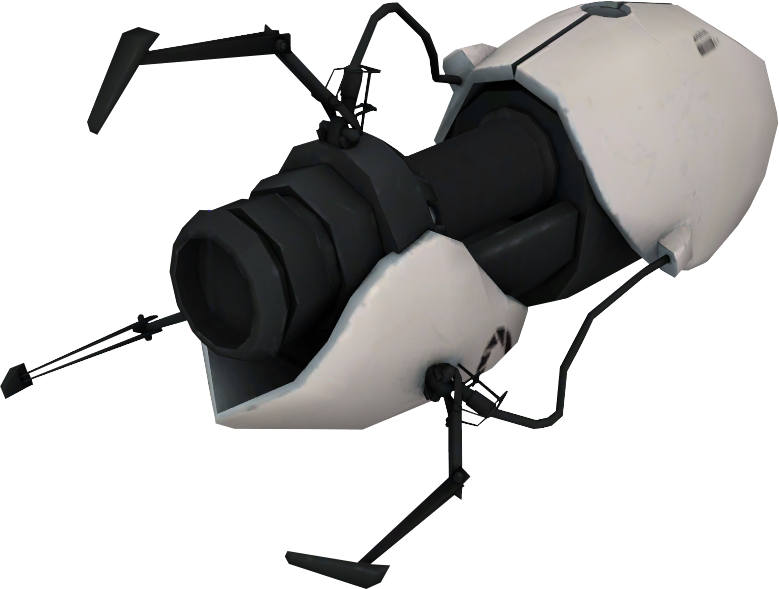}
        \captionof*{figure}{Portal Gun}
        \label{fig:PortalGun}

        \end{minipage}
\end{center}

\begin{center}
    \captionsetup{justification=centering}
    \begin{minipage}{.68\textwidth}
    \centering
    \begin{enumerate}
\setcounter{enumi}{\the\value{papercount}}
 \item A \emph{High Energy Pellet} (HEP) is a spherical object which moves in a straight line until it encounters another object. HEPs move faster than the player avatar. If they collide with the player avatar, then the avatar is killed. If a HEP encounters a wall or another object, it will bounce off it with equal angle of incidence and reflection. In Portal, some HEPs have a finite lifespan, which is reset when the HEP passes through a portal, and others have an unbounded lifespan. These unbounded HEPs are referred to as \emph{Super High Energy Pellets}. \ifabstract HEP's are created by \emph{HEP Launchers} and can activate \emph{HEP Catchers} if they come in contact with them.\fi
 \setcounter{papercount}{\the\value{enumi}}
\end{enumerate}
        \end{minipage}
\hfill  
      \begin{minipage}{0.28\textwidth}
    \centering
            \includegraphics[width=0.9\linewidth]{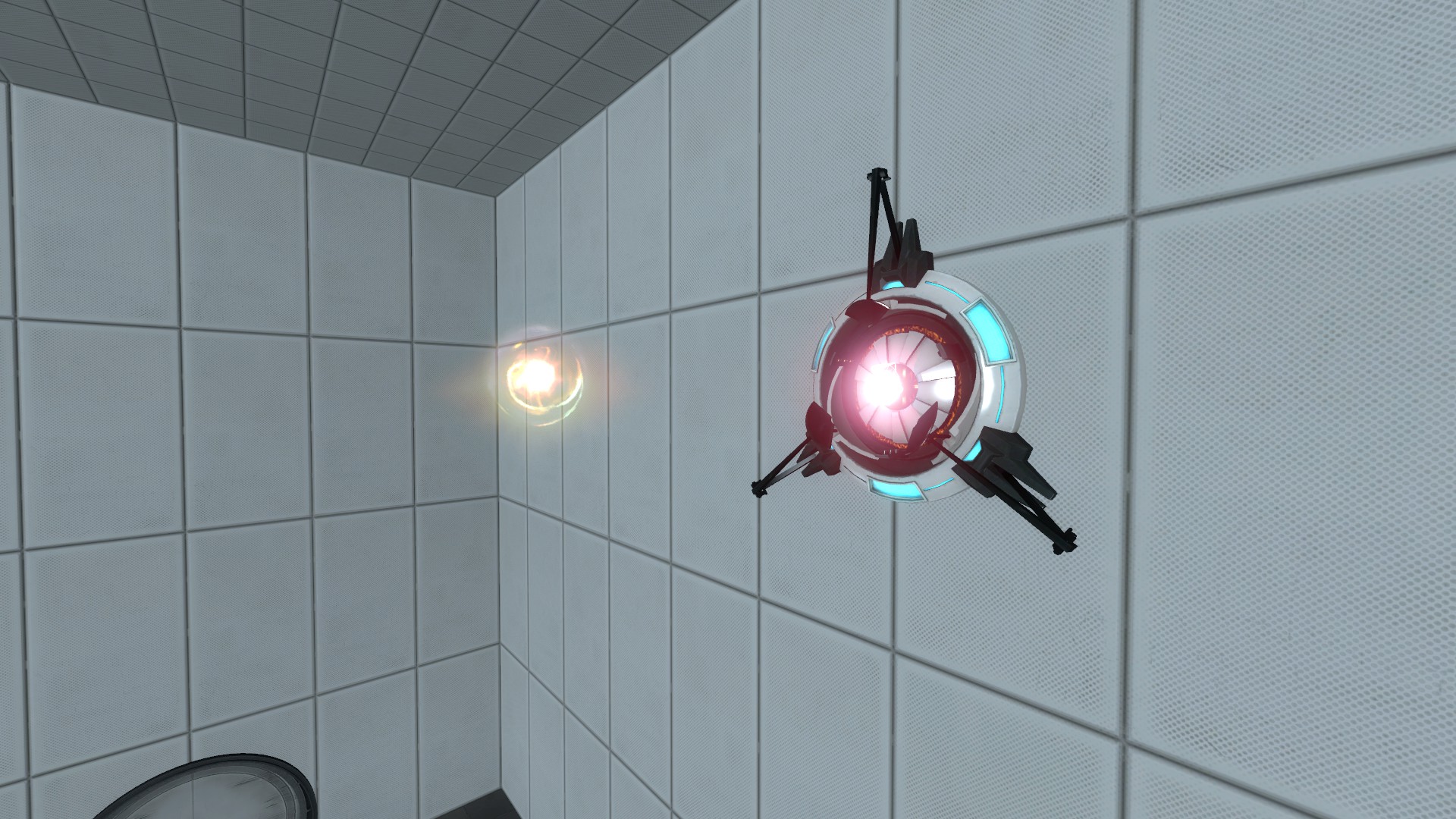}
        \captionof*{figure}{A HEP about to reach a HEP Collector}
        \label{fig:HEP}

        \end{minipage}
\end{center}
\iffull
\begin{center}
    \captionsetup{justification=centering}
    \begin{minipage}{.68\textwidth}
    \centering
    \begin{enumerate}
\setcounter{enumi}{\the\value{papercount}}
 \item A \emph{HEP Launcher} emits a HEP at an angle normal to the surface upon which it is placed. These are launched when the HEP launcher is activated or when the previously emitted HEP has been destroyed.
\setcounter{papercount}{\the\value{enumi}}
\end{enumerate}
        \end{minipage}
        \hfill
    \begin{minipage}{0.28\textwidth}
    \centering
 \includegraphics[height=0.11\textheight]{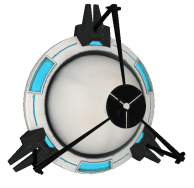}
        \captionof*{figure}{HEP Launcher}
        \label{fig:HEPLauncher}

        \end{minipage}
\end{center}

\begin{center}
    \captionsetup{justification=centering}
    \begin{minipage}{.68\textwidth}
    \centering
    \begin{enumerate}
\setcounter{enumi}{\the\value{papercount}}
\item A \emph{HEP Catcher} is a device which is activated if it is ever hit by a HEP. In Portal, this device can act as a button, and is commonly used to open doors or move platforms when activated.
\setcounter{papercount}{\the\value{enumi}}
\end{enumerate}
        \end{minipage}
        \hfill
    \begin{minipage}{0.28\textwidth}
    \centering
\includegraphics[height=0.11\textheight]{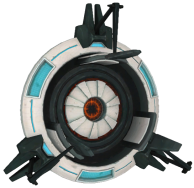}
        \captionof*{figure}{HEP Catcher}
        \label{fig:HEPCatcher}

        \end{minipage}
\end{center}
\fi
\begin{center}
    \captionsetup{justification=centering}
    \begin{minipage}{.68\textwidth}
    \centering
    \begin{enumerate}
\setcounter{enumi}{\the\value{papercount}}
\item A \emph{Laser Emitter} emits a \emph{Thermal Discouragement Beam} at an angle normal to the surface upon which it is placed. The beam travels in a straight line until it is stopped by a wall or another object. The beam causes damage to the player avatar and will kill the avatar if they stay close to it for too long. We call the beam and its emitter a \emph{laser}.
\setcounter{papercount}{\the\value{enumi}}
\end{enumerate}
                \end{minipage}
        \hfill
    \begin{minipage}{0.28\textwidth}
    \centering
\includegraphics[height=0.11\textheight]{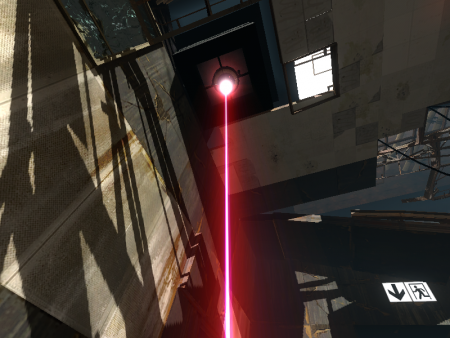}
        \captionof*{figure}{A Laser Emitter and Thermal Discouragement Beam.}
        \label{fig:laser}

        \end{minipage}
\end{center}

\begin{center}
    \captionsetup{justification=centering}
    \begin{minipage}{.68\textwidth}
    \centering
    \begin{enumerate}
\setcounter{enumi}{\the\value{papercount}}
\item A \emph{Laser Relay} is an object which can activate other objects while a laser passes through it.
\vspace{2mm}
 \item A \emph{Laser Catcher} is an object which can activate other objects while a contacts it.
\setcounter{papercount}{\the\value{enumi}}
\end{enumerate}
        \end{minipage}
        \hfill
    \begin{minipage}{0.28\textwidth}
    \centering
    \vspace{4mm}
\includegraphics[height=0.11\textheight]{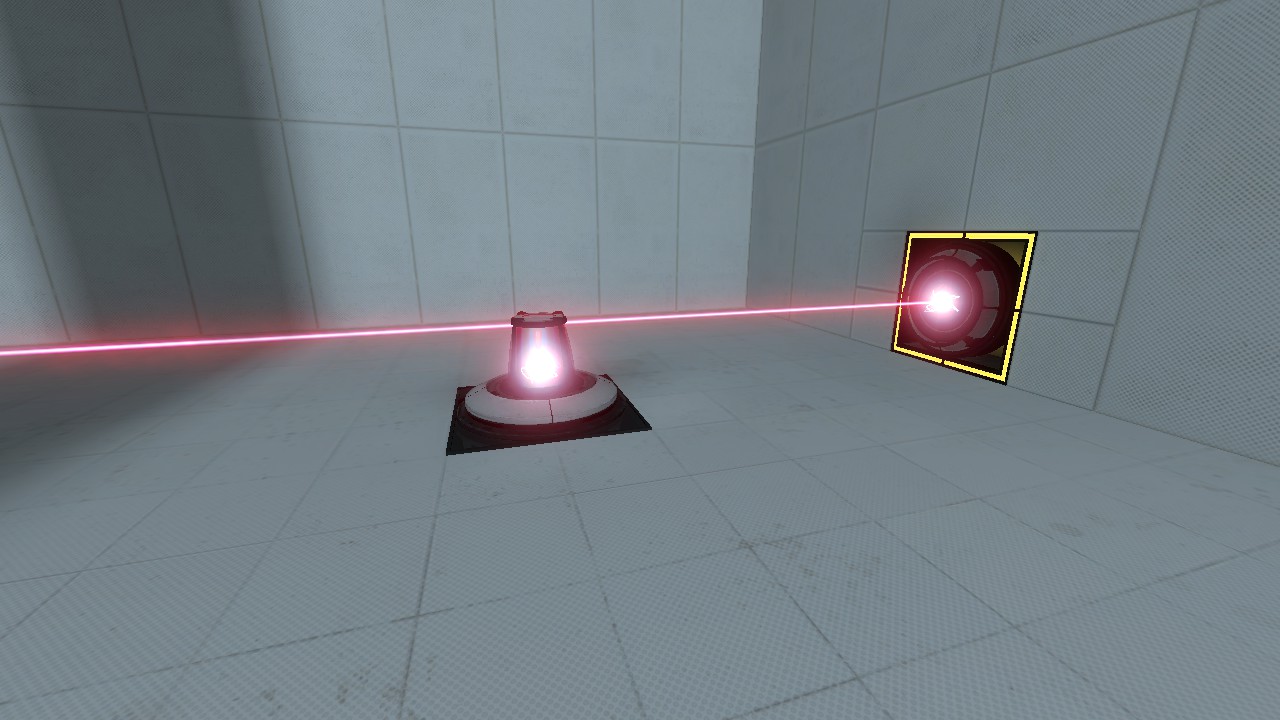}
        \captionof*{figure}{An active laser relay and laser catcher.}
        \label{fig:Laser}
        \end{minipage}

\end{center}

\begin{center}
    \captionsetup{justification=centering}
    \begin{minipage}{.68\textwidth}
    \centering
    \begin{enumerate}
\setcounter{enumi}{\the\value{papercount}}
\item A \emph{Moving Platform} is a solid polygon with an inactive and an active position. It begins in the inactive position and will move in a line at a constant velocity to the active position when activated. If it becomes deactivated it will move back to the inactive position with the opposite velocity.
\setcounter{papercount}{\the\value{enumi}}
\end{enumerate}
        \end{minipage}
        \hfill
    \begin{minipage}{0.28\textwidth}
    \centering
    \vspace{4mm}
\includegraphics[height=0.11\textheight]{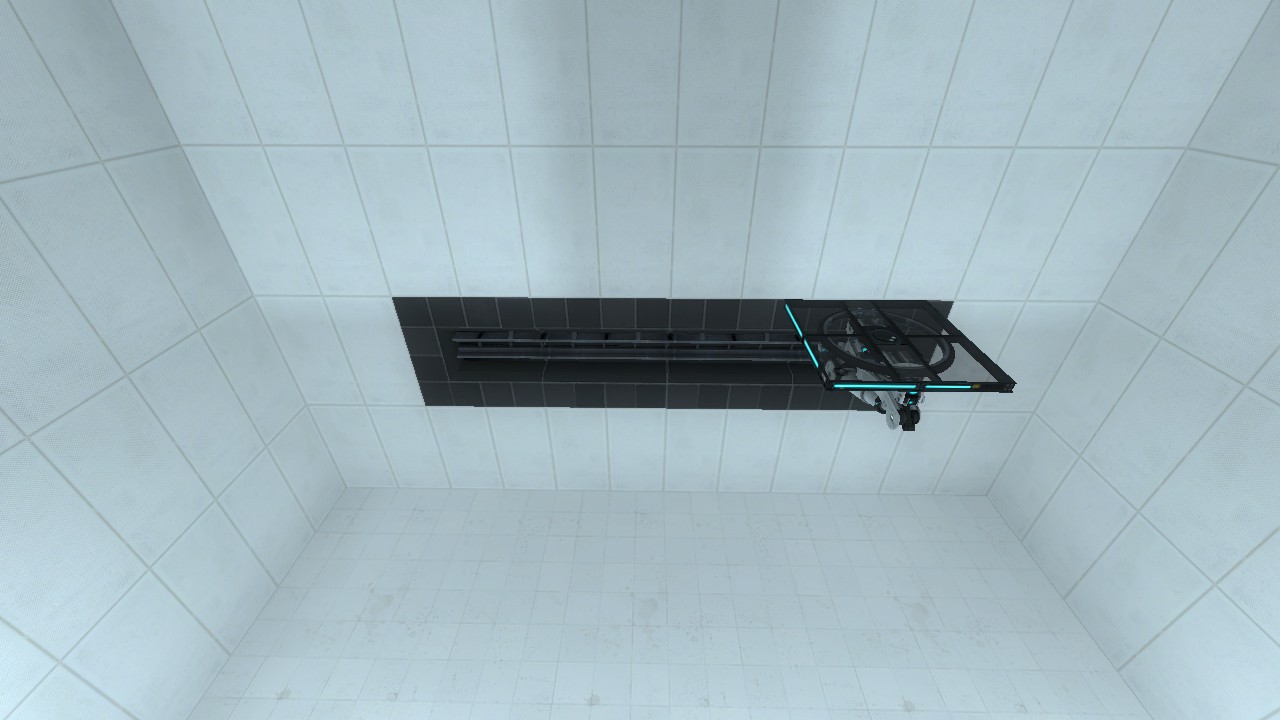}
        \captionof*{figure}{A horizontal moving platform.}
        \label{fig:MovingPlatform}

        \end{minipage}
\end{center}

\begin{center}
    \captionsetup{justification=centering}
    \begin{minipage}{.68\textwidth}
    \centering
    \begin{enumerate}
\setcounter{enumi}{\the\value{papercount}}
\item A \emph{Turret} is an enemy which cannot move on its own. If the player's avatar is within the field of view of a turret, the turret will fire on the avatar. If the avatar is shot sufficiently many times within a short period of time, the avatar will die.
\setcounter{papercount}{\the\value{enumi}}
\end{enumerate}
        \end{minipage}
        \hfill
    \begin{minipage}{0.28\textwidth}
    \centering
\includegraphics[height=0.11\textheight]{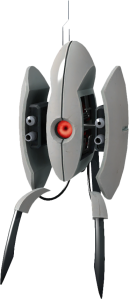}
        \captionof*{figure}{Turret from Portal 2}
        \label{fig:Turret}

        \end{minipage}
\end{center}

\begin{center}
    \captionsetup{justification=centering}
    \begin{minipage}{.68\textwidth}
    \centering
    \begin{enumerate}
\setcounter{enumi}{\the\value{papercount}}
\item An \emph{Excursion Funnel}, also called a \emph{Gravitational Beam Emitter} emits a gravitational beam normal to the surface upon which it is placed. The gravitational beam is directed and will move small objects at a constant velocity in the prescribed direction. Importantly, it will carry Weighted Storage Cubes and the player avatar. Gravitational Beam Emitters can be switched on and off, as well as flipping the direction of the gravitational beam they emit.
\setcounter{papercount}{\the\value{enumi}}
\end{enumerate}
        \end{minipage}
        \hfill
    \begin{minipage}{0.28\textwidth}
    \centering
\includegraphics[height=0.11\textheight]{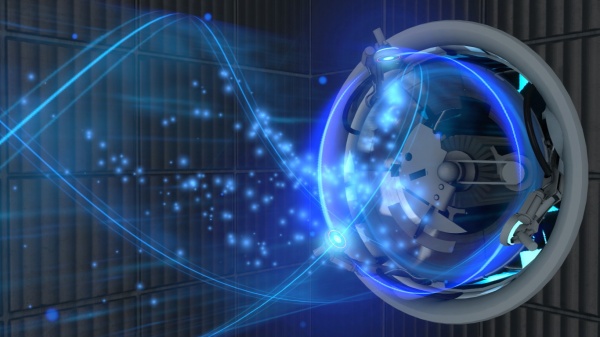}
        \captionof*{figure}{A Gravity Beam and Excursion Funnel.}
        \label{fig:ExcursionFunnel}

        \end{minipage}
\end{center}

\later{
There are two main pieces of software for creating levels in Portal 2: the \emph{Puzzle Maker} (also known as the \emph{Puzzle Creator}), and the \emph{Valve Hammer Editor} equipped with the \emph{Portal 2 Authoring Tools}. Both of these tools are publicly available for players to create their own levels. The Puzzle Maker is a more restricted editor than Hammer, with the advantage of providing a more user-friendly editing experience. 
However, levels created in the Puzzle Maker must be coarsely discretized, with coarsely discretized object locations, and must be made of voxels. In particular, the Puzzle Maker uses the P2C file format, which restricts it to pseudopolynomial instances (while Hammer uses VMF). Furthermore, no HEP launchers or additional doors can be placed in Puzzle Maker levels. We will often comment on which of our reductions can be constructed with the additional Puzzle Maker restrictions (except, of course, the small level size and item count), but this distinction is not a primary focus of this work.}

\section{Movement is Easy}
\label{sec:PortalNavigation}
\ifabstract\later{
\section{Movement is Easy}
\label{appen:PortalNavigation}}
\fi

In this section, we prove a basic result that the core mechanism of portals
does not affect the complexity of traversing a level.

\begin{theorem}
\label{thm:pseudopoly}
\textsc{Portal} with portals can be solved in pseudopolynomial time.
\end{theorem}
\begin{proof}
We construct a state-space graph of the Portal level. Each vertex represents a tuple comprised of the avatar's position vector, the avatar's velocity vector, the avatar's orientation, the position vector of the blue portal, and the position vector of the orange portal. The vertices are connected with directed edges encoding the state transitions caused by user input. We can then search for a path from the initial game state to any of the winning game states in time polynomial in the size of the graph.

Thus we have a pseudopolynomial-time algorithm for solving \textsc{Portal} in this case.
\end{proof}

\section{Portal with Emancipation Grills is Weakly NP-hard}
\label{sec:PortalFalling}

In this section, we prove that \textsc{Portal} with portals and Emancipation Grills is weakly NP-hard by reduction from \textsc{Subset Sum} \cite{NPBook}, which is defined like so.
\begin{problem}
\textsc{Subset Sum}

\textit{Input:} A set of integers $A=\{a_1,a_2,\dots a_n\}$, and a target value $t$.

\textit{Output:} Whether there exists a subset $\{s_1,s_2,\dots,s_m\}\subseteq A$ such that 
\begin{align*}
\sum_{i=1}^{m}s_i=t.
\end{align*}
\end{problem}
The reduction involves representing the integers in $A$ as distances which are translated into the avatar's velocity. More explicitly, the input $A$ will be constructed from long holes the avatar can fall down, and the target will be encoded in a distance the avatar must launch themselves after falling. In the game, there is a maximum velocity the player avatar can reach. For the next theorem, it is necessary to consider  Portal without bounded terminal velocity.\footnote{Alternatively, any terminal velocity which scales at least polynomially in the level size suffices.}

\begin{theorem}
\textsc{Portal} with portals, long fall, Emacipation Grills, and no terminal velocity is weakly NP-hard.
\end{theorem}

\begin{wrapfigure}{R}{0.5\textwidth}
  \centering
  \vspace{-2ex}
  \includegraphics[width=\linewidth]{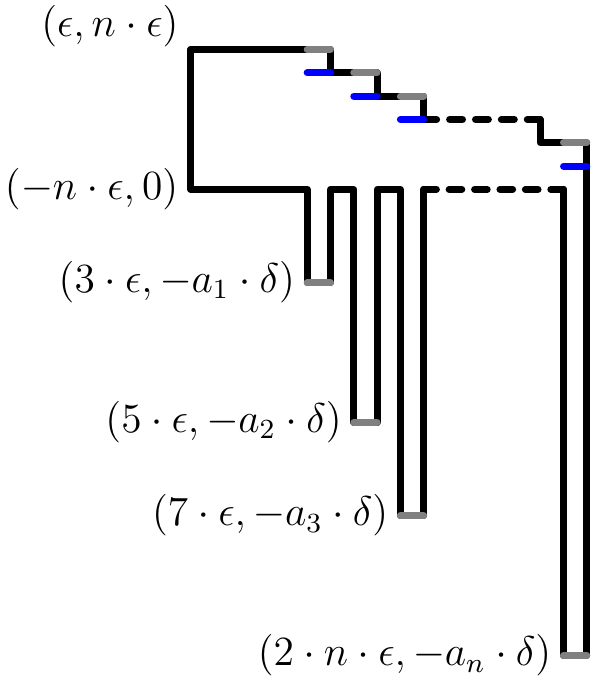}
  \caption{A cross-section of the element selection gadget, where $\delta=2\cdot n^2\cdot \epsilon \cdot t$. Grey lines are portalable surfaces and blue lines are Emancipation Grills.}
  \vspace{-4ex}
  \label{fig:fallingDiagram}
\end{wrapfigure}

\begin{proof}
The elements of $A$ are represented by a series of wells, each of depth $4\cdot a_i\cdot n^2\cdot \epsilon\cdot t$, where $a_i\in A$ is the number to be encoded, $n$ is the number of elements in $A$, $t$ is the original target value of the \textsc{Subset Sum} problem, and $\epsilon$ is an expansion factor that is chosen to larger than the height of the avatar plus the height she can jump. An example is shown in Figure~\ref{fig:fallingDiagram}. The bottom of each well is a portalable surface, and the ceiling above each well is also a portalable surface. This construction will allow the avatar to shoot a portal to a ceiling tile, and to the bottom of the well they are falling into, selecting the next number. 

We cannot allow the avatar to select the same element more than once. The Emancipation Grills below each portalable ceiling serve to remove the portal from the ceiling of the well into which the avatar is currently falling, and to prevent sending a portal up to that same ceiling tile. The stair-stepped ceiling will allow the player to see the ceilings of all of the wells with index greater than the one they are currently at, but prevents them from seeing the portalable surface of the wells with a lower index. This construction ensures that the player can only select each element once using portals. The enforced order of choosing does not matter when solving \textsc{Subset Sum}.

Another concern is the ability to move horizontally while falling. This movement is a small, fixed velocity $v_h$. To solve this issue, we simply ensure the distance between each hole is greater than $2\cdot v_h\cdot  n\cdot  \epsilon$ so it is impossible to move from one hole to another while falling.

The distance between each step is $\epsilon$, thus to ensure the accumulated error from falling these distances does not impact the solution to the subset sum, we scale each $a_i$ by $4\cdot n^2\cdot\epsilon$, which is greater than the sum of all of these extra distances.

The verification gadget involves two main pieces: a single portalable surface on a wall, the launching portal, and a target platform for the player to reach. We place the launching portal so it can always be shot from the region above the wells. The target platform is placed $\epsilon$ units below the launching portal. The target platform is placed a distance of $2\cdot t\cdot n\cdot \epsilon$ away from the launching portal and in front of the portalable surface such that leaving the portalable surface with the target velocity will cause the player to reach the target platform. Because it takes 1 second to fall the vertical distance to the platform, the avatar will only reach the target if their velocity is equal to $2 \cdot t\cdot n^2\cdot \epsilon$. We make the target platform $n\cdot \epsilon$ on each side, to account for any errors incurred by the falling region or initial horizontal movement. This size is smaller than the difference if the target value $t$ differed by $1$. We now have an encoding of our numbers, a method of selecting them, and verification if they reach the target sum, completing the reduction.

With an acceleration of $\alpha$ and zero initial velocity a body will fall a distance $s = \frac{1}{2}\cdot\alpha\cdot t^2$. The time it takes to fall will thus be $t = \left(\frac{2s}{\alpha}\right)^{1/2}$. The resulting final velocity will be $v_f=\left(2\cdot\alpha s\right)^{1/2}$. If the player starts at an initial height $h$ and horizontal velocity $V_x$ then they will travel a total horizontal distance of $V_x\cdot\left(\frac{2h}{\alpha}\right)^{1/2}$. In our construction we have the player initially fall a total distance of 
\begin{align*}\sum_{i=1}^{k} 4\cdot s_i\cdot n^2\cdot \epsilon\cdot t\end{align*} 
If this solution is correct, the sum of the $a_j$ which are chosen will add to $t$ giving $2\cdot n^2 \cdot \epsilon\cdot t^2$. Because the verification portal on the wall is placed at a height of $\epsilon$ we arrive at our required distance to the verification platform of $t' = \left(8\cdot\alpha\cdot\epsilon\cdot n^2\cdot t^2\right)^{1/2}\left(\frac{2\epsilon}{\alpha}\right)^{1/2} = 2\cdot n\cdot \epsilon\cdot t$.
\end{proof}

\iffull
All of the game elements needed for this construction can be placed in the Puzzle Maker. However, this reduction would not be constructible because maps in the Puzzle Maker appear to be specified in terms of voxels. Because \textsc{Subset Sum} is only weakly NP-hard\cite{NPBook}, we need the values of the elements of $A$ to be exponential in $n$. Thus we need to describe the map in terms of coordinates specifying the polygons making up the map, whereas the Puzzle Maker specifies each voxel in the map.
\fi
\begin{corollary}
\textsc{Portal} with portals, long fall, and no terminal velocity can be solved in pseudopolynomial time.
\end{corollary}
\begin{proof}
Theorem~\ref{thm:pseudopoly} gives a pseudopolynomial time algorithm for \textsc{Portal} with portals by constructing the full state-space graph. The state of the emancipation grids do not get changed over time and thus do not add additional state that needs to be stored. \iffull We can use the same vertices in the former proof, but now the edge transitions will differ if they player's avatar passes through any emancipation grids. This construction is still polynomial in the state-space and thus polynomial in the voxels in the level.\fi
\end{proof}

\section{Portal with Turrets is NP-hard}
\label{sec:PortalTurrets}
\ifabstract\later{
\section{Portal with Turrets is NP-hard}
\label{appen:PortalTurrets}}
\fi

In this section we prove \textsc{Portal} with turrets is NP-hard, and show that our method can be generalised to prove that many 3D platform games with enemies are NP-hard. Although enemies in a game can provide interesting and complex interactions, we can pull out a few simple properties that will allow them to be used as gadgets to reduce solving a game from 3-SAT, defined like so.
\begin{problem}
$3$-\textsc{SAT}

\textit{Input:} A $3$-CNF boolean formula $f$.

\textit{Output:} Whether there exists a satisfying assignment for $f$.
\end{problem}

This proof follows the architecture laid out in \cite{NintendoFun2014}:
\begin{enumerate}
\item The enemy must be able to prevent the player from traversing a specific region of the map; call this the \emph{blocked region}.
\item The player avatar must be able to enter an area of the map, which is path-disconnected from the blocked region, but from which the player can remove the enemy in the blocked region.
\item The level must contain long falls.
\end{enumerate}

We further assume that the behavior of the enemies is local, meaning an interaction with one enemy will not effect the behavior of another enemy if they are sufficiently far away. \iffull In many games one must also be careful about ammo and any damage the player may incur while interacting with the gadget, because these quantities will scale with the number of literals. Here long falls serve only in the construction of one-way gadgets, and can of course be replaced by some equivalent game mechanic. Similarly, a 2D game with these elements and an appropriate crossover gadget should also be NP-hard. \fi The following is a construction proving Portal with Turrets is NP-hard using this technique. Note that these gadgets can be constructed in the Portal 2 Puzzle Maker.
\iffull
\subsection{Literal}\fi
Each literal is encoded with a hallway with three turrents placed in a raised section, illustrated in Figure \ref{TurretLiteral}. The hallway must be traversed by the player, starting from ``Traverse In'', ending at ``Traverse Out''. If the turrets are active, they will kill the avatar before the avatar can cross the hallway or reach the turrets. The literal is true if the turrets are deactivated or removed, and false if they are active. The ``Unlock In'' and ``Unlock Out'' pathways allow for the player avatar to destroy the turrets from behind, deactivating them and counting as a true assignment of the literal.
\iffull
\begin{figure}[th]
  \centering
    \includegraphics[width=0.8\textwidth]{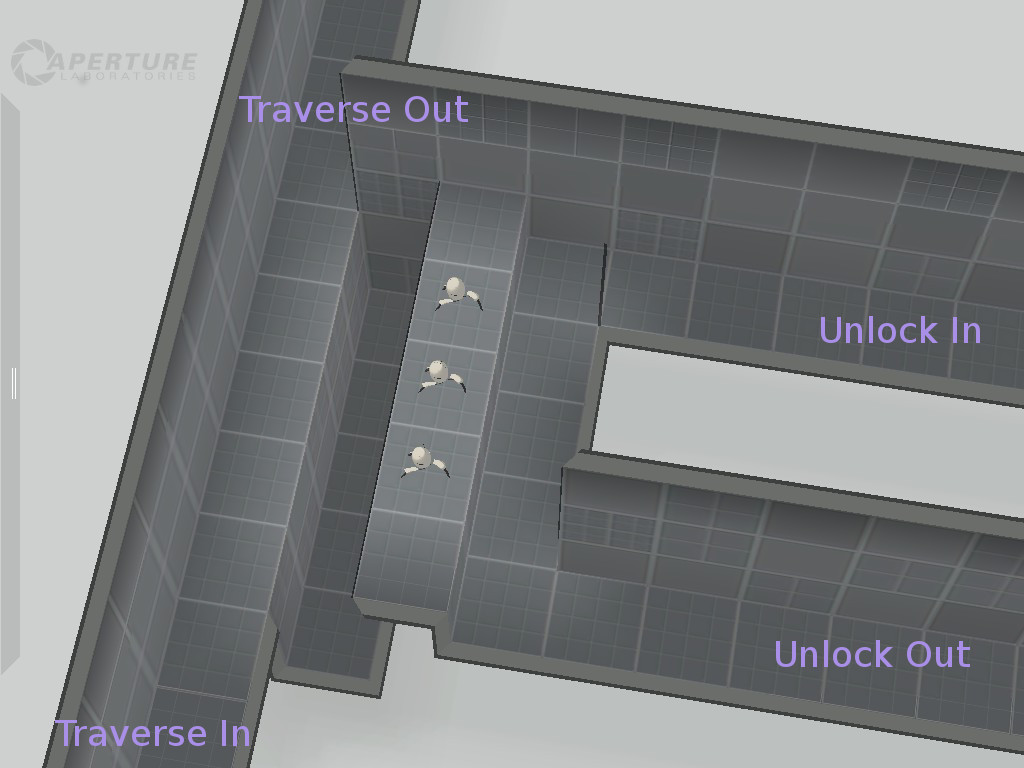}
    \caption{An example of a (currently) false literal constructed with Turrets. Labels added over the screenshot denote }
    \label{TurretLiteral}
\end{figure}
\fi
\iffull
\subsection{Variable}\fi
The variable gadget consists of a hallway that splits into two separate paths. Each hallway starts and ends with a one-way gadget constructed with a long fall. This construction forces the avatar to commit to one of the two paths. The gadget is shown in Figure~\ref{Split}. 
The hallways connect the ``Unlock In'' and ``Unlock Out'' paths of the literals corresponding to a particular variable. Furthermore, one path connects all of the true literals, the other connects all of the false literals.
\later{
\begin{figure}[!ht]
  \centering
    \includegraphics[width=.8\textwidth]{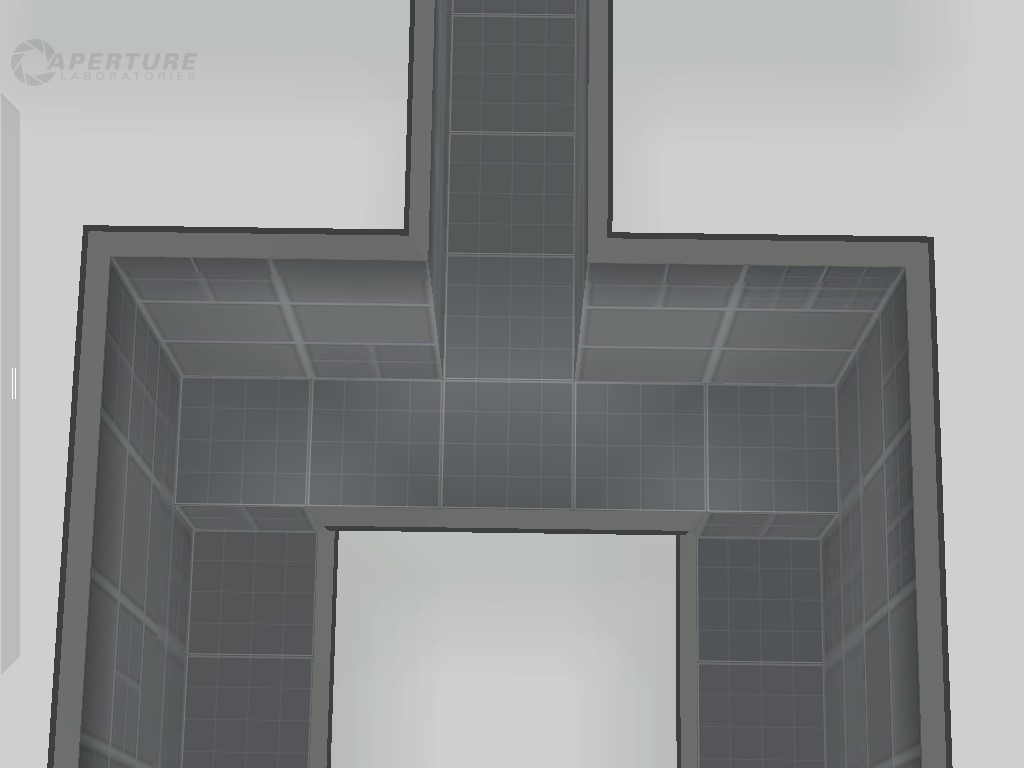}
    \caption{An example of the choice gadget used to construct variable gadgets.}
    \label{Split}
\end{figure}
}
\iffull
\subsection{Clause Gadget}\fi
Each clause gadget is implemented with three hallways in parallel. A section of each hallway is the ``Traverse In'' through the ``Traverse Out'' corresponding to a literal. The avatar can progress from one end of the clause to the other if any of the literals is true (and thus passable). Furthermore, each of the clause gadgets is connected in series. Figures \ref{TurretClauseGadget} and \ref{TurretClauseExample} illustrate a full clause gadget.
\begin{figure}[th]
  \centering
    \includegraphics[width=0.8\textwidth]{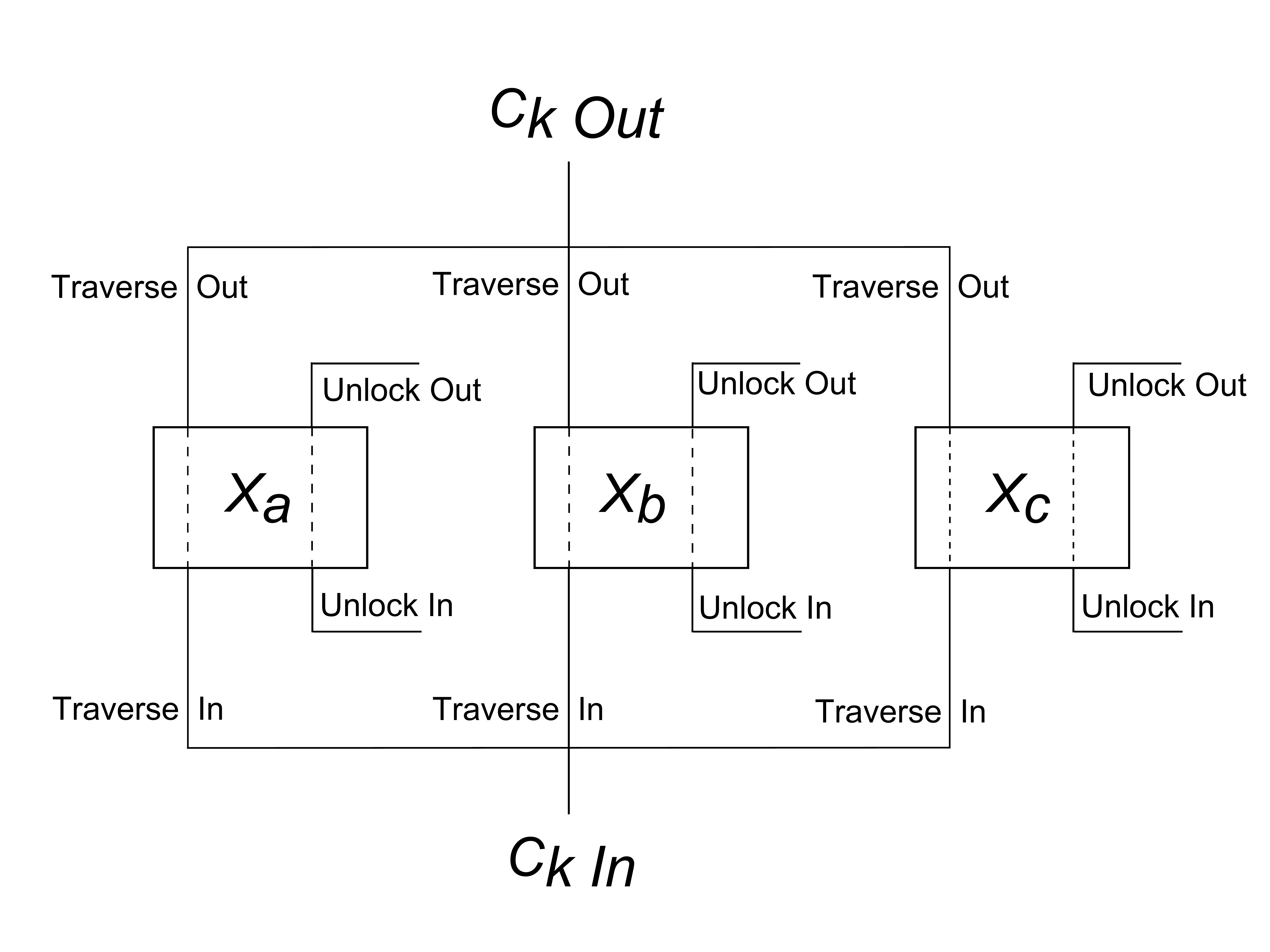}
    \caption{A diagram of clause $C_k$ which contains variables $x_a$, $x_b$, and $x_c$.}
    \label{TurretClauseGadget}
\end{figure}

\begin{figure}[th]
  \centering
    \includegraphics[width=0.6\textwidth]{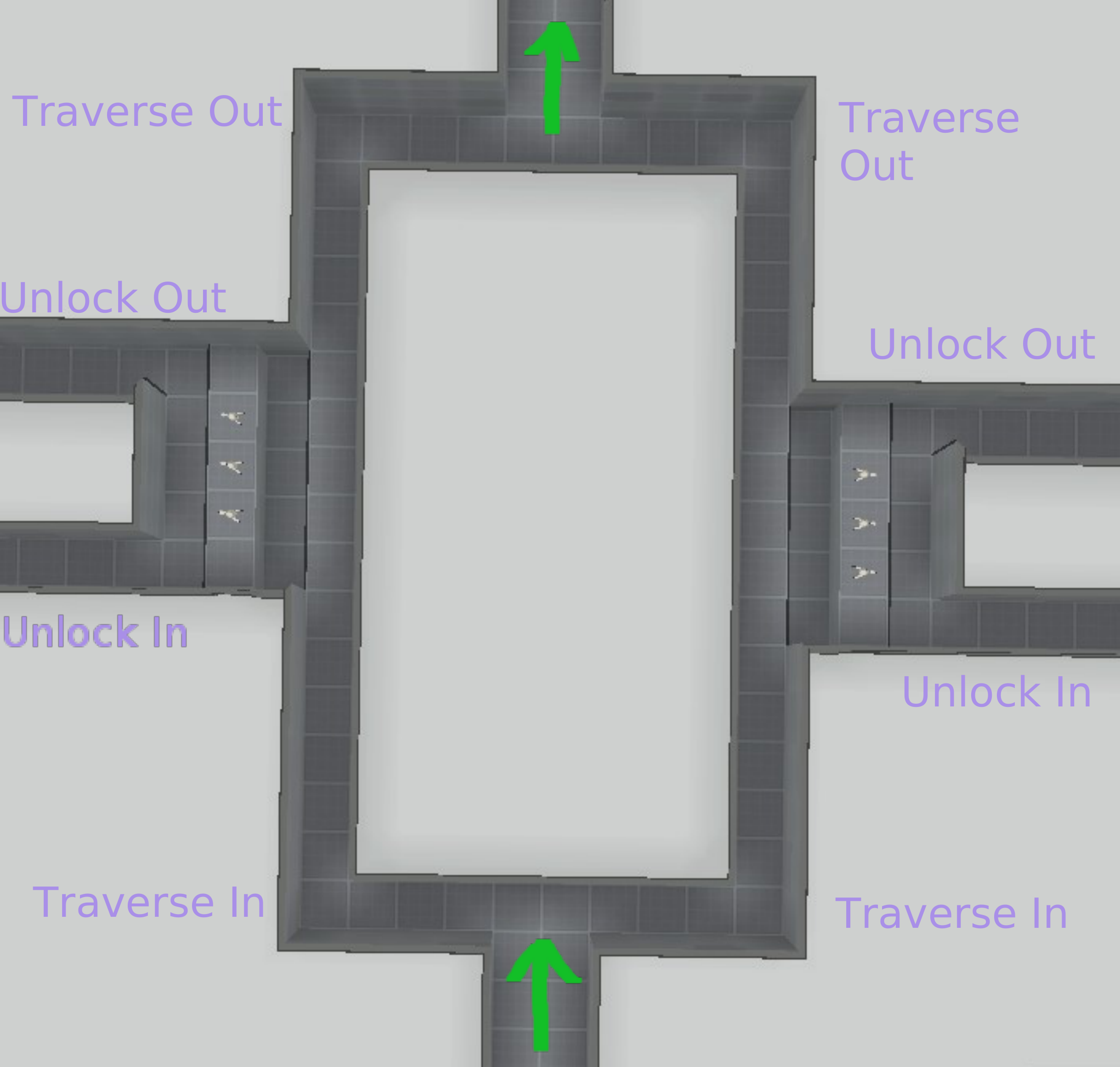}
    \caption{An example of a clause gadget with two literals.}
    \label{TurretClauseExample}
\end{figure}

\ifabstract
\begin{figure}[th]
  \centering
    \includegraphics[width=0.8\textwidth]{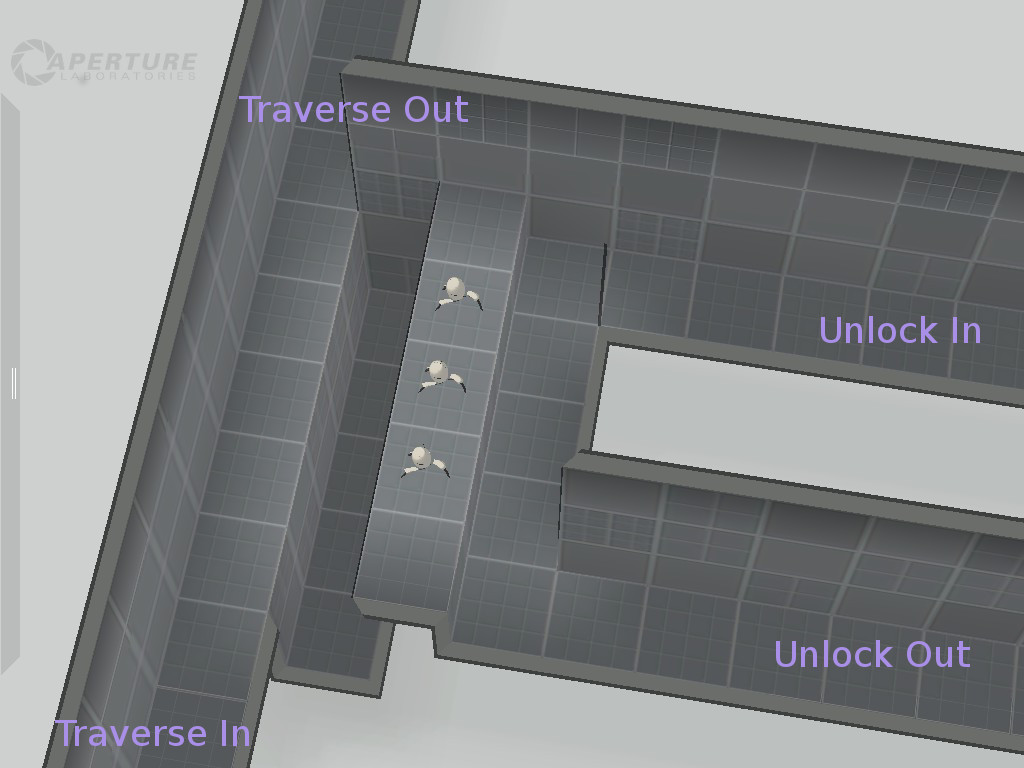}
    \caption{An example of a (currently) false literal constructed with Turrets. Labels added over the screenshot.}
    \label{TurretLiteral}
\end{figure}
\fi
\both{
\begin{theorem}
  \textsc{Portal} with Turrets is NP-hard.
\end{theorem}}
\later{
\begin{proof}
Given an instance of a 3SAT problem, we can translate it into a Portal with Turrets map using the above gadgets. This map is solvable if and only if the corresponding 3SAT problem is solvable.
\end{proof}

It is tempting to claim NP-completeness because disabling the turrets need only be performed once per turret and thus seems to have a monotonically changing state. However, the turrets themselves are physical objects that can be picked up and moved around. Their relocation add an exponential amount of state to the level. Further, if they can be jumped on top of or used to block the player in a constrained hallway, they may conceivably cause the level to be PSPACE-complete in the same way boxes can add significant complexity to a game.
}

\iffull
\subsection{Application to Other Games}
\label{subsec:OtherGames}\fi
While the framework we have presented is shown using the gameplay elements of Portal, similar elements to those we have used show up in other video games. Hence, our framework can be generalised to show hardness of other games. In this section we note several common features of games which would allow for an equivalent to the turret ``guarding unit'' in Portal. We list examples of notable games which fit the criteria. We give ideas how to use our framework to prove hardness results for these games, but it is important to note that game-specific implementation details will need to be taken into account for any hardness proof.

The first examples are games that include player controlled weapons with fixed positions, such as stationary turrets or gun emplacements. The immovable turrets should be placed at the unlock points of the literal gadget, so that they only allow the player to shoot the one desired blocking unit. Examples in contemporary video games include the Emplacement Gun in Half-Life 2, the Type-26 ASG in Half-Life, and the Anti-Infantry Stationary Guns in Halo 1 through 4.

Another set of examples are games which include a pair of ranged weapons, where one is more powerful than the other, but has shorter range. In place of the turrets in the Portal literal gadgets, we place an enemy unit equipped with the short range weapon, and give the player avatar the long range weapon. We place the blocked region such that it is in range and line of sight of the player while standing in the unlock region of the literal gadget. Additionally, we place the player such that they are not in range of the enemy's weapon. Thus the player can kill the enemy from the unlock area.
Suppose further that the blocked region is built in such a way that the player can only pass through it by moving within range of the enemy. One way of doing this would be to build it with tight turns. The result would be an equivalent implementation of the variable and clause gadgets from our Portal constructions. Note that a special case involves melee enemies. This construction applies to Doom, the Elder Scrolls III--V, Fallout 3 and 4, Grand Theft Auto 3--5, Left 4 Dead 1 and 2, the Mass Effect series, the Deus Ex series, the Metal Gear Solid series, the Resident Evil series, and many others. \iffull The complementary case occurs when the player has the short ranged, but more powerful weapon and the enemy has the weaker, long ranged weapon. Here the unlock region provides close proximity to the enemy unit but the locked region involves a significant region within line of sight and range of the enemy but is outside of the player's weapon's range. Although most games where this construction is applicable will also fall into the prior case, examples exist where the player has limited attacks, such as in the Spyro series.\fi

\iffull
A third case is where the environment impacts the effectiveness of attacks. For example, certain barriers might block projectile weapons but not magic spells. Skills that can shoot above or around barriers like this show up with Thunderstorm in Diablo II, Firestorm in Guild Wars, and Psi-storm in StarCraft. Another common effect is a location based bonus, for example the elevated-ground bonus in XCOM. Unfortunately these games lack a long-fall, and thus require the construction of a one-way gadget if one wishes to prove hardness.

While we have so far only covered NP-hardness, we conjecture that these games are significantly harder. 
Assuming simple AI and perfect information, many are likely PSPACE-complete; however, when all of the details are taken into consideration, EXPTIME or NEXPTIME seem more likely. Proving such results will require development of more sophisticated mathematical machinery.
\fi

\section{Portal with Timed Door Buttons is NP-hard}
\label{sec:PortalTimed}
\ifabstract\later{
\section{Portal with Timed Door Buttons is NP-hard}
\label{appen:PortalTimed}}
\fi

We provide a new metatheorem related to Forisek's Metatheorem 2~\cite{Forisek10} and Viglietta's Metatheorem 1~\cite{HardGames12}.
\begin{metatheorem}
\label{thm:timed-thm}
  A platform game with doors controlled by timed switches is NP-hard.
\end{metatheorem}
\begin{proof}
\label{pf:timed-proof}
We will prove hardness by reducing from finding Hamiltonian cycles in grid graphs~\cite{GridHamPath}. Every vertex of the graph will be represented by a room with a timed switch in the middle. These rooms will be laid out in a grid with hallways in-between. The rooms are small in comparison to the hallways. In particular, the time it takes to press a timed button and travel across a room is $\delta$ and the time it takes to traverse a hallway is $\alpha > n\cdot\delta$ where $n$ is the number of nodes in the graph. This property ensures the error from turning versus going straight through a room won't matter in comparison to traveling from node to node. All of the timed switches will be connected to a series of closed doors blocking the exit hallway connected to the start node. The timers will be set, such that the doors will close again after $(\alpha + \delta) \cdot (t + 1) + \epsilon$ where $\epsilon$ is the time it takes to move from the switch at the start node through the open doors to the exit. The exit is thus only reachable if all of the timed switches are simultaneously active. Because we can make $\alpha$ much larger than $\epsilon$, we can ensure that there is only time to visit every switch exactly once and then pass through before any of the doors revert.
\end{proof}
\begin{corollary}
\label{cor:timeddoor}
A Portal level with only timed door buttons is NP-hard.
\end{corollary}

\later{
A screenshot of an example map for Corollary~\ref{cor:timeddoor} is given in Figure~\ref{fig:HampathScreenshotsMap}. Because the Portal 2 Workshop does not allow additional doors, the example uses collapsible stairs, as seen in Figure~\ref{fig:HampathScreenshotsVerify} for the verification gadget instead. We note that anything which will prevent the player from passing unless currently activated by a timed button will suffice. Moving platforms and Laser Fields are other examples. Unfortunately, the Puzzle Maker does not allow the timer length to be specified, which is a needed generalization for the reduction and available in the Hammer editor.

\begin{figure}[!ht]
  \centering
    \includegraphics[width=0.8\textwidth]{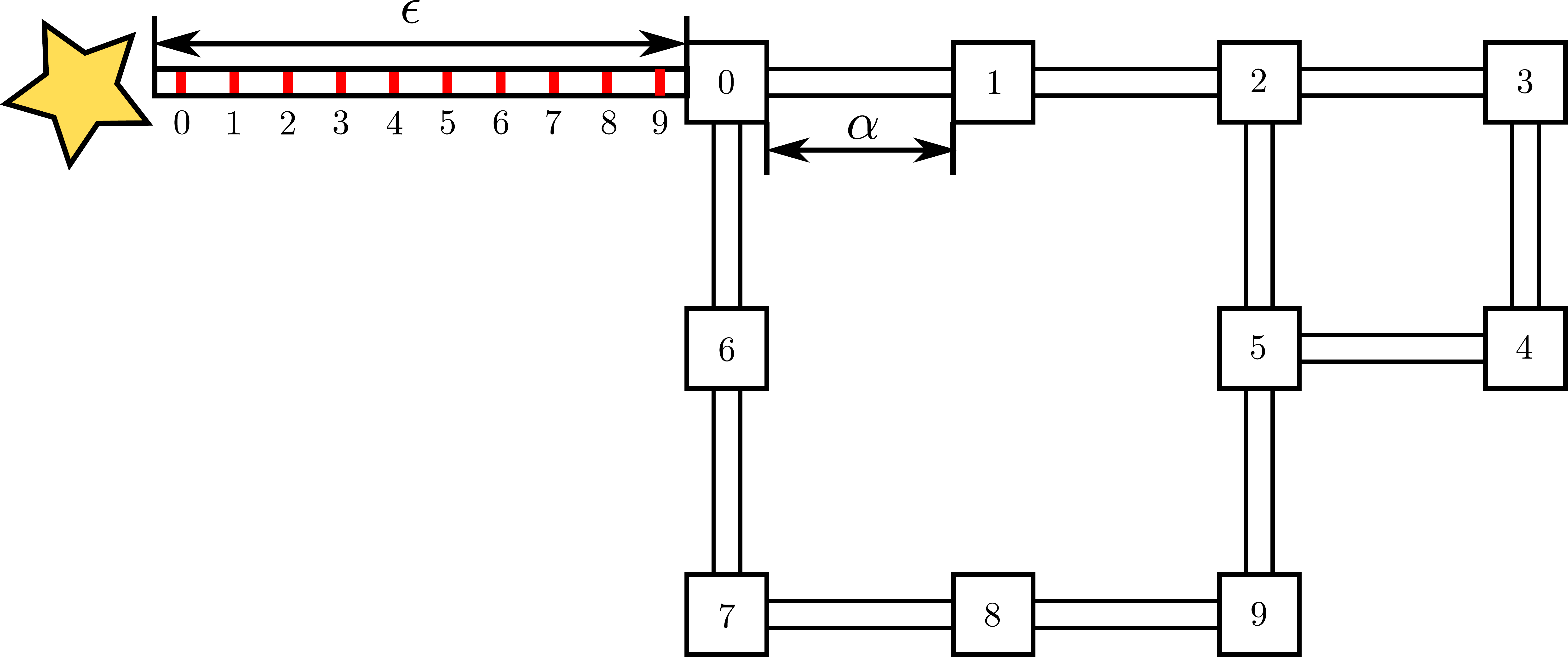}
    \caption{An example of a map forcing the player to find a Hamiltonian cycle in a grid graph.}
  \label{fig:HampathScreenshotsMap}
\end{figure}

\begin{figure}[!ht]
  \centering
    \includegraphics[width=0.8\textwidth]{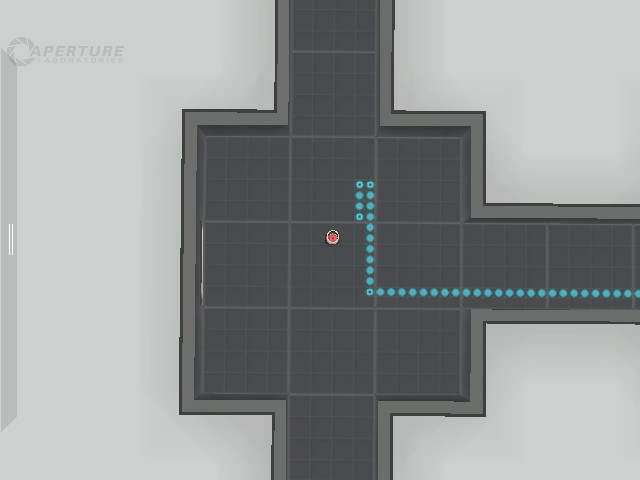}
    \caption{Close-up of a node in the grid graph.}
  \label{HampathScreenshotsNode}
\end{figure}

\begin{figure}[!ht]
  \centering
    \includegraphics[width=0.8\textwidth]{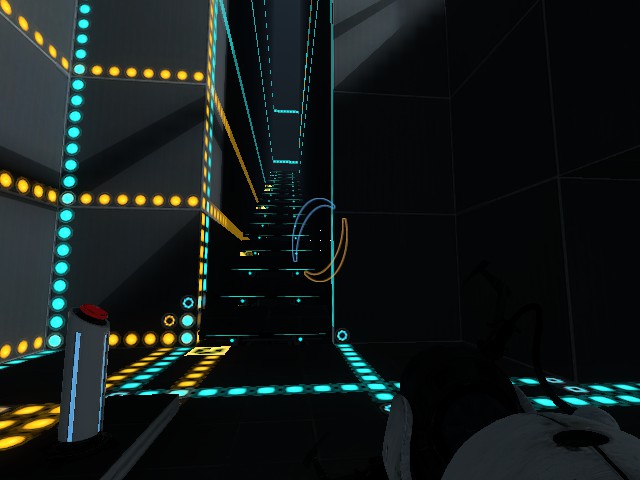}
    \caption{A screenshot of the verification gadget, partially satisfied.}
  \label{fig:HampathScreenshotsVerify}
\end{figure}
}

\section{Portal with High Energy Pellets and Portals is NP-hard}
\label{sec:PortalHEP}

\ifabstract\later{
\section{Portal with High Energy Pellets and Portals is NP-hard}
\label{appen:PortalHEP}}
\fi

In Portal, the High Energy Pellet, HEP, is an object which moves in a straight line until it encounters another object. HEPs move faster than the player avatar and if they collide with the player avatar, the avatar is killed. If a HEP encounters another wall or object, it will bounce off of that object with equal angle of incidence and reflection. In Portal, some HEPs have a finite lifespan, which is reset when the HEP passes through a portal, and others have an unbounded lifespan. \iffull A HEP launcher emits a HEP normal to the surface it is placed upon. These are launched when the HEP launcher is activated or when the previous HEP emitted has been destroyed.\fi A HEP catcher is another device that is activated if it is ever hit by a HEP. When activated this device can activate other objects, such as doors or moving platforms. HEP's are only seen in the first Portal game and are not present in the Portal 2 Puzzle Maker.

\begin{theorem}
\textsc{Portal} with Portals, High Energy Pellets, HEP launchers, HEP catchers, and doors controlled by HEP catchers is NP-hard.
\end{theorem}
\begin{proof}
We will reduce from finding Hamiltonian cycles in grid graphs~\cite{GridHamPath}; refer to Figure~\ref{fig:doorsHamiltonianHEP}. For this construction, we will need a gadget to ensure the avatar traverses every represented node, as well as a timing element. Each node in the graph will be represented by a room that contains a HEP launcher and a HEP catcher. They are positioned near the ceiling, each facing a portalable surface. The HEP catcher is connected to a closed door preventing the avatar from reaching the exit. Proof of \iffull Metatheorem~\ref{thm:timed-thm} uses the same idea and has an example of how rooms in Portal can be connected to simulate a grid graph. The rooms are small in comparison to the hallways. In particular, the time it takes to shoot a portal, wait for it to enter the HEP Catcher, and travel across a room is $\delta$ and the time it takes to traverse a hallway is $\alpha > n\cdot\delta$ where $n$ is the number of nodes in the graph. This property ensures the error from turning versus going straight through a room won't matter in comparison to traveling from node to node.\fi

\begin{figure}[t]
  \centering
    \includegraphics[width=0.72\textwidth]{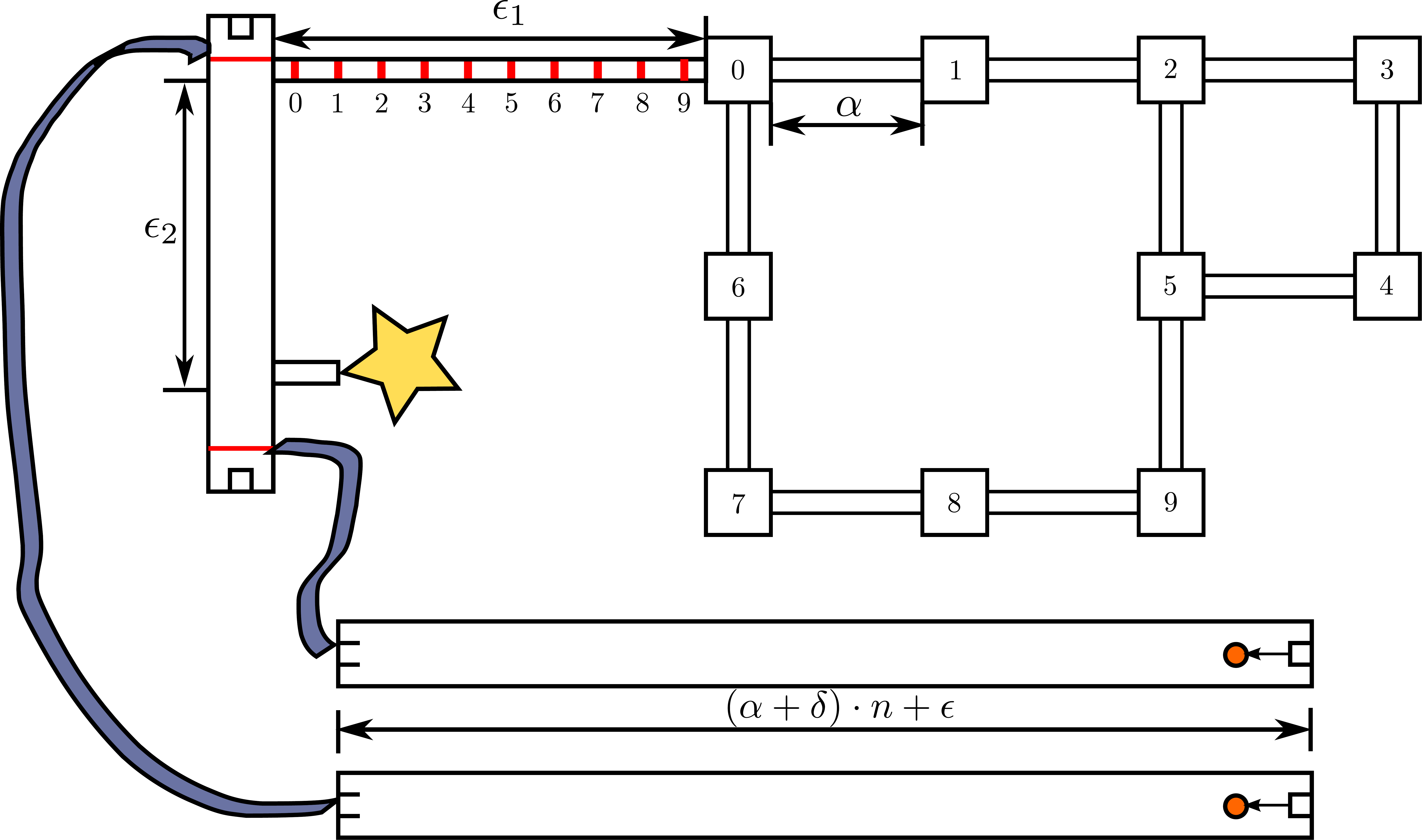}
    \caption{An example level for the HEP reduction. Not drawn to scale.}
    \label{fig:doorsHamiltonianHEP}
\end{figure}

The timer will contain two elements. First, we will arrange for a hallway with two exits and a HEP launcher behind a door on one end. The hallway is long enough so it is impossible for the avatar to traverse the hallway when the door is open. Call this component the \emph{time verifier}. In another area, we have a HEP launcher and a HEP catcher on opposite ends of a hallway that is inaccessible to the avatar. The catcher in this section will open the door in the time verifier. This construction ensures that the player can only pass through the time verifier if they enter it before a certain point after starting. To complete the proof, we set the timer equal to $(\alpha+\delta)\cdot n+ \epsilon_1+\epsilon_2$ where $\epsilon_1$ is the minimum time needed for the avatar to traverse the hallway with doors, $\epsilon_2$ is the minimum time needed for the avatar to traverse the time verifier, $\alpha$ is the minimum time it takes for the player to move to an adjacent room and change the trajectory of the HEP, and $n + 1$ is the number of HEP catchers in the level. Thus concludes our reduction from the Hamiltonian cycle problem in grid graphs.
\end{proof}

The HEP Catchers are only able to be activated once, so one may be tempted to claim this problem is in NP. This is not necessarily the case because navigating around HEP particles with more complicated trajectories might require long paths or wait times. The PSPACE-hardness of motion planning with periodic obstacles\cite{sutner1988motion} suggests the natural class for this problem is actually PSPACE-complete.

\section{Portal is PSPACE-complete}
\label{sec:pspace}

In this section we give a new metatheorem for games with doors and switches, in the same vein as the metatheorems in \cite{Forisek10}, \cite{HardGames12}, and \cite{2Button2015}. We use this metatheorem to give proofs of PSPACE-completeness of Portal with various game elements. All of the gadgets in this section can be created in the Portal 2 Puzzle Maker.

The proofs in this section revolve around constructing game mechanics which implement a switch:
the construction can be in one of two states, and the state is controllable by the player. When the avatar is near the switch, it can be freely set to either state. Each state has a set of doors which are open when the switch is in that state. A switch is very similar to a button in that it controls whether doors are open or closed, and the player has the option of interacting with it. The key difference is that buttons can be pressed multiple times to open or close its associated doors, and cannot necessarily be `unpressed' to undo the action. We show that a game with switches and doors is PSPACE-complete, using similar techniques to \cite{2Button2015}.

In what follows we will use the nondeterministic constraint logic framework\cite{GPCBook09}, wherein the state of a nondeterministic machine is encoded by a graph called a \emph{constraint graph}. The state is updated by changing the orientation of the edges in such a way that constraints stored on the vertices are satisfied.

Formally, an constraint graph is an undirected simple graph $G=(V,E)$ with an assignment of nonnegative integers to the edges $w:E\rightarrow \mathbb{Z}^+$, referred to as \emph{weights}, and an assignment of integers to the vertices $c:V\rightarrow \mathbb{Z}$, referred to as \emph{constraints}. Each edge has an orientation $p:E\rightarrow \{+1,-1\}$. A constraint graph is fully specified by the tuple $\mathcal{G}=(G,w,c,p)$. The edge orientation $p$ induces a directed graph $D_{G,p}$.  
Let $v\in V$ be a vertex of $G$. Its \emph{in-neighbourhood} 
\begin{align*}
N^-(v,p)=\{w~|~(v,w)\in A\}
\end{align*}
is the set of vertices of $D_{G,p}=(V,A)$ with an arc oriented towards it.
The constraint graph $\mathcal{G}$ is \emph{valid} if, for all $y\in V$, $\sum_{x\in N^-(y,p)} w((x,y)) \ge c(x)$.
The state of a constraint graph can be changed by selecting an edge and multiplying its orientation by $-1$, such that the resulting constraint graph is valid. We say that we have \emph{flipped} the edge.

A vertex $v$ in a constraint graph with three incident edges $x,y,o$ can implement an AND gate by setting $c(v)=2$, $w(x)=w(y)=1$, and $w(o)=2$. Clearly, the edge $o$ can only point away from $v$ if both $x$ and $y$ are pointing towards $v$. In a similar fashion, we can implement an OR gate by setting $w(v)=2$, $w(x)=w(y)=w(o)=2$. A constraint graph where all vertices are AND or OR vertices is called an \emph{AND/OR constraint graph}. The following decision problem about constraint graphs is PSPACE-complete.

\begin{problem}
\textsc{Nondeterministic Constraint Logic}

\textit{Input}: An AND/OR constraint logic graph $\mathcal{G}=((V,E),w,c,p)$, and a target edge ${i,j}\in E$.

\textit{Output}: Whether there exists a constraint graph $\mathcal{G}'=((V,E),w,c,p')$ such that $p'(\{i,j\})=-p(\{i,j\})$, and which can be obtained from $\mathcal{G}$ by a sequence of valid edge flips.
\end{problem}
\begin{metatheorem}
\label{thm:switches}
Games with doors that can be controlled by a single switch and switches that can control at least six doors are PSPACE-complete.
\end{metatheorem}
\begin{proof}
We prove this by reduction from \textsc{Nondeterministic Constraint Logic}. The edges of the consistency graph are represented by a single switch whose state represents the edge orientation. Connected to each switch is a \emph{consistency check gadget}. This gadget consists of a series of hallways that checks that the state of the two vertices adjacent to the simulated edge are in a valid configuration and thus that the update made to the graph was valid. Each edge switch is connected to doors in up to six consistency checks, two for itself and four for the adjacent edges. For an AND vertex, the weight two edge is given by the door with the single hallway, and the weight one edges connect to the two doors in the other hallway. For an OR vertex we have a hallway that splits in three, each with one node. An example is given in Figure~\ref{img:SwitchGadget}. Each switch thus connects to five doors. All of the edge gadgets, with their constraint checks, are connected together. This construction allows the player to change the direction of any edge they choose. However, to get back to the main hallway connecting the gadgets, the graph must be left in a valid state. Off the main hallway there is be a final exit connected to the target location, but blocked by a door connected to the target edge. If the player is able to flip the edge by visiting the edge gadget, moving the cube to the button which opens the exit door, and return through the graph consistency check, then the avatar can reach the target location.
\end{proof}

\begin{figure}[h]
    \centering
    \begin{subfigure}[b]{0.45\textwidth}
        \centering
        \includegraphics[width = \textwidth]{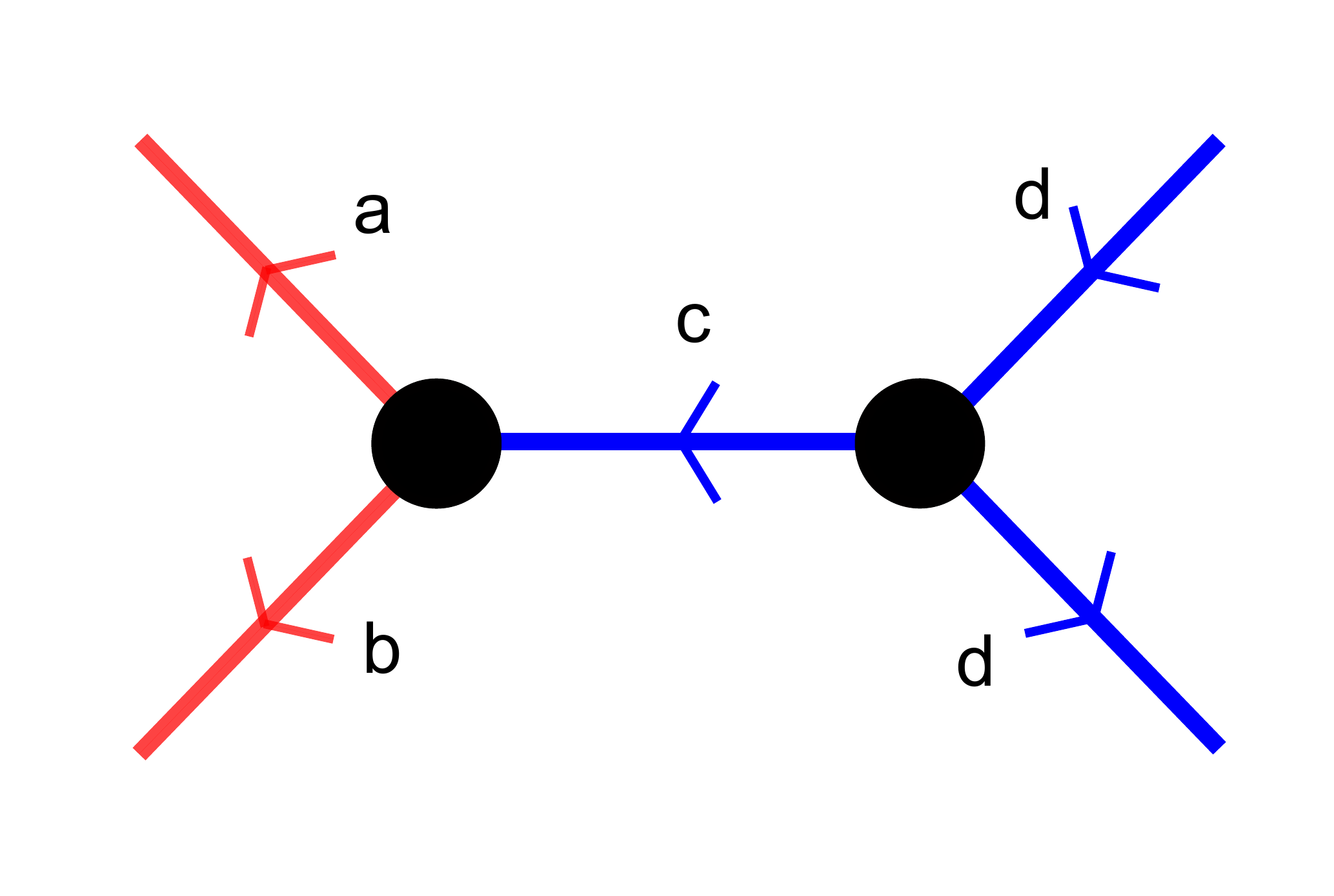}
        \caption{Section of a constraint logic graph being simulated. Blue edges are weight 2 and red edges are weight 1.}
    \end{subfigure}
    \hfill
    \begin{subfigure}[b]{0.45\textwidth}
        \centering
        \includegraphics[width = \textwidth]{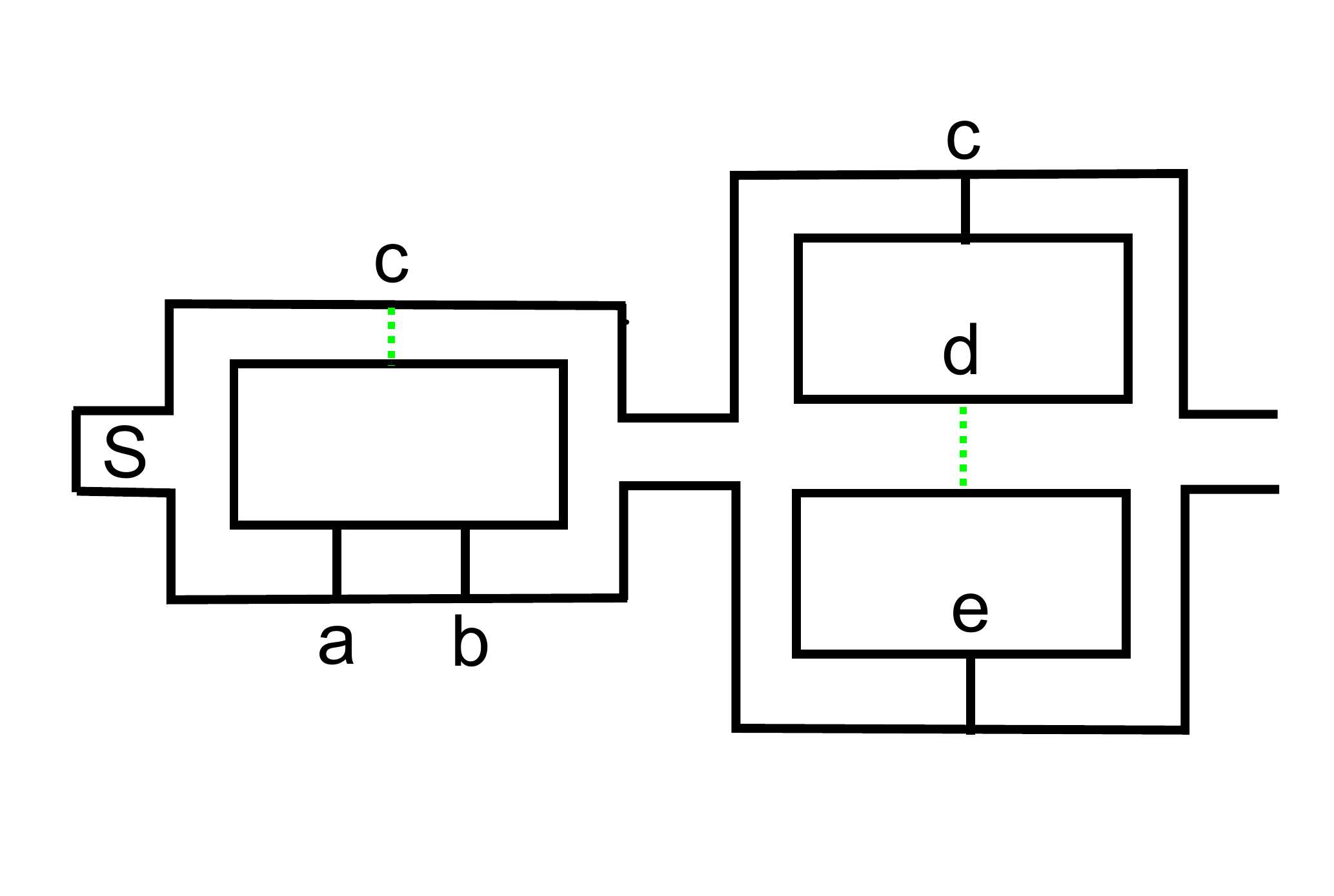}
        \caption{Gadget simulating edge $c$ in the constraint logic graph. Green dotted lines are open doors.}
    \end{subfigure}
    \caption{Example of an edge gadget built from switches and doors.}
\label{img:SwitchGadget}
\end{figure}

\begin{theorem}
\label{thm:pspace}
 \textsc{Portal} with any subset of long falls, portals, Weighted Storage Cubes, doors, Heavy Duty Super Buttons, lasers, laser relays, gravity beams, turrets, timed buttons, and moving platforms is in PSPACE.
\end{theorem}
\begin{proof}
Portal levels do not increase in size and the walls and floors have a fixed geometry. Assuming all velocities are polynomially bounded, all gameplay elements have a polynomial amount of state which describes them. For example the position and velocity of the avatar or a HEP; whether a door is open or closed; and the time on a button timer. The number of gameplay elements remains bounded while playing. Most gameplay elements cannot be added while playing, and items like the HEP launcher and cube suppliers only produce another copy when the prior one has been destroyed. We only need a polynomial amount of space to describe the state of a game of Portal at any given point in time. Thus one can nondeterministically search the state space for any solutions to the \textsc{Portal} problem, putting it in NPSPACE. Thus by Savitch's Theorem\cite{SAVITCH1970} the problem is in PSPACE.
\end{proof}

\begin{theorem}\label{thm:cubes-pspace}
 \textsc{Portal} with Weighted Storage Cubes, doors,  and Heavy Duty Super Buttons is PSPACE-complete.
\end{theorem}
\begin{proof}
We will construct switches and doors out of doors, Weighted Storage Cubes,  and Heavy Duty Super Buttons. Then, we invoke Metatheorem~\ref{thm:switches} to complete the proof. A switch is constructed out of a room with a single cube and two buttons as in Figure~\ref{fig:switch}. Which of the buttons being pressed by the cube dictates the state of the switch. Each button is connected to the corresponding doors which should open when the switch is in that state. To ensure the switch is always in a valid state, we put an additional door in the only entrance to the room. This door is only open if at least one of the two buttons is depressed. Furthermore, this construction prevents the cube from being removed from the room to be used elsewhere. As long as there are no extra cubes in the level, the room must be left in exactly one of the two valid switch states for the avatar to exit the room. We now apply our doors and simulated switches as in Metatheorem~\ref{thm:switches} completing the hardness proof. Theorem~\ref{thm:pspace} implies inclusion in PSPACE.

\begin{figure}[h]
  \centering
    \includegraphics[width=0.8\textwidth]{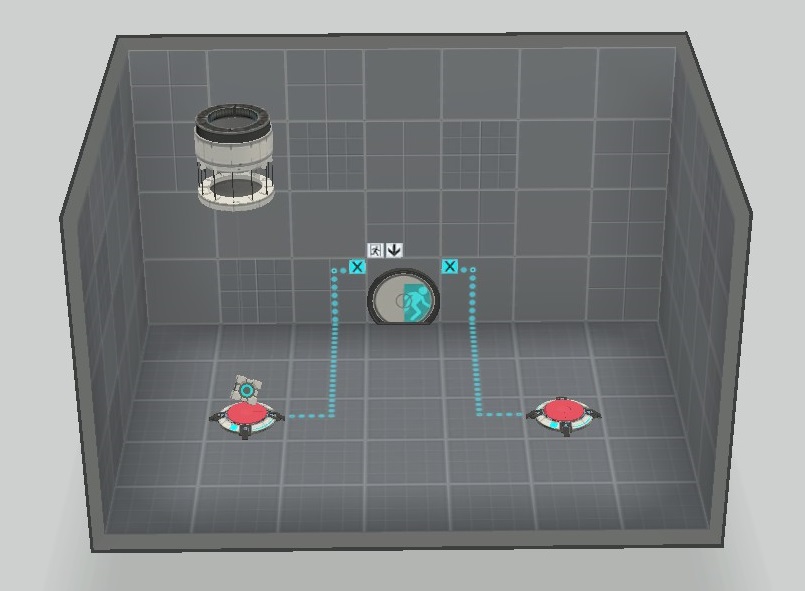}
    \caption{An example of a single switch implemented with cubes, doors, and buttons. The door will only open if at least one of the buttons is pressed.}
    \label{fig:switch}
\end{figure}
\end{proof}

\begin{theorem} \label{thm:lasers-pspace}
\textsc{Portal} with lasers, relays, portals, and moving platforms is PSPACE-complete.
\end{theorem}
\begin{proof}
We will construct doors and switches out of lasers, relays, and moving platforms allowing us to use Metatheorem~\ref{thm:switches}. In Portal 2, the avatar is not able to cross through an active laser. Because lasers can be blocked by the moving platforms game element, a door can be constructed by placing a moving platform and laser at one end of a small hallway. If the moving platform is in front of the laser, the gadget is in the unlocked state. If the moving platform is to the side, then the player cannot pass through the hallway and it is in the locked state. Moving platforms can be controlled by laser relays and will switch position based on whether the laser relay is active. Lasers can be directed to selectively activate laser relays with portals, so we have a mechanism to lock or unlock the doors.

As it stands, once a new portal is created the previously opened door will revert to its previous state. To prove PSPACE-hardness, we need to make these changes persist. To do so, we introduce a memory latch gadget, shown in Figures~\ref{fig:memory-latch0}\iffull~and \ref{fig:memory-latch1}\fi. When the relay in this gadget is activated for a sufficiently long period of time, the platform will move out of the way and the laser will keep the relay active. If the relay has been blocked for enough time, the platform moves back and blocks the laser. Thus, the state of the gadget persists.

\begin{figure}[h]
  \centering
    \includegraphics[width=0.8\textwidth]{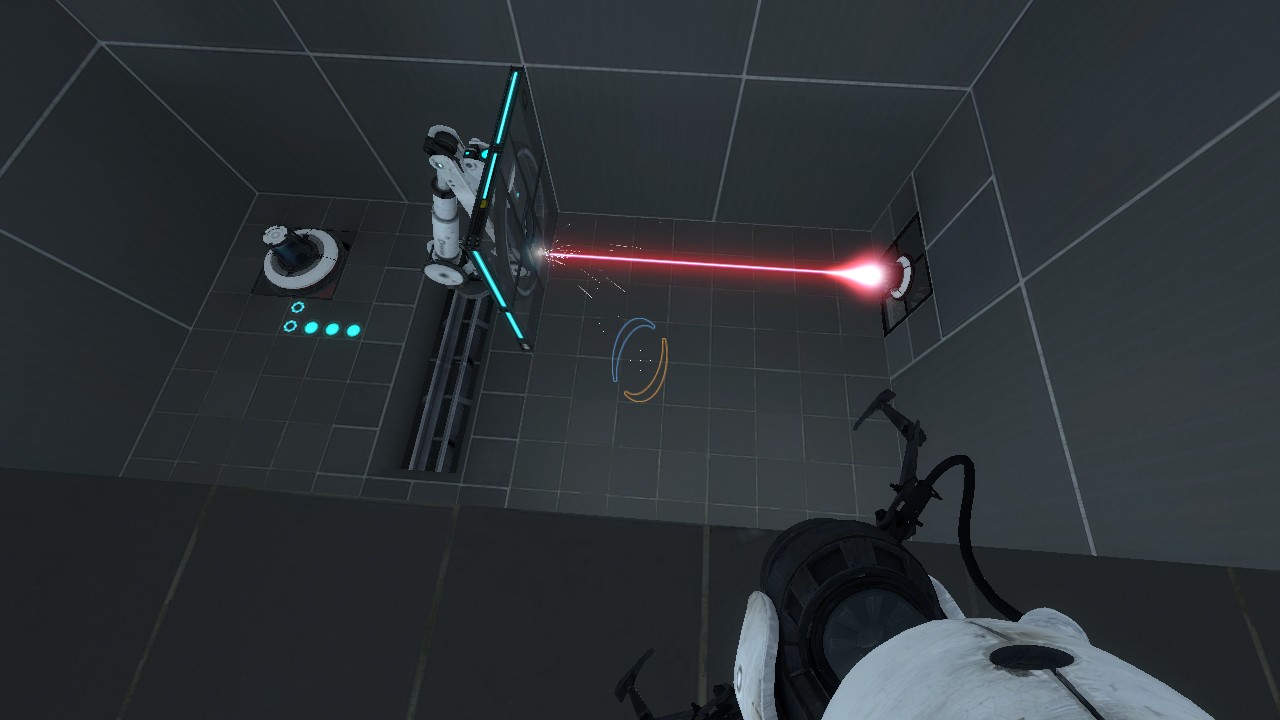}
    \caption{A memory latch in the off state.}
    \label{fig:memory-latch0}
\end{figure}

\begin{figure}[h]
  \centering
    \includegraphics[width=0.8\textwidth]{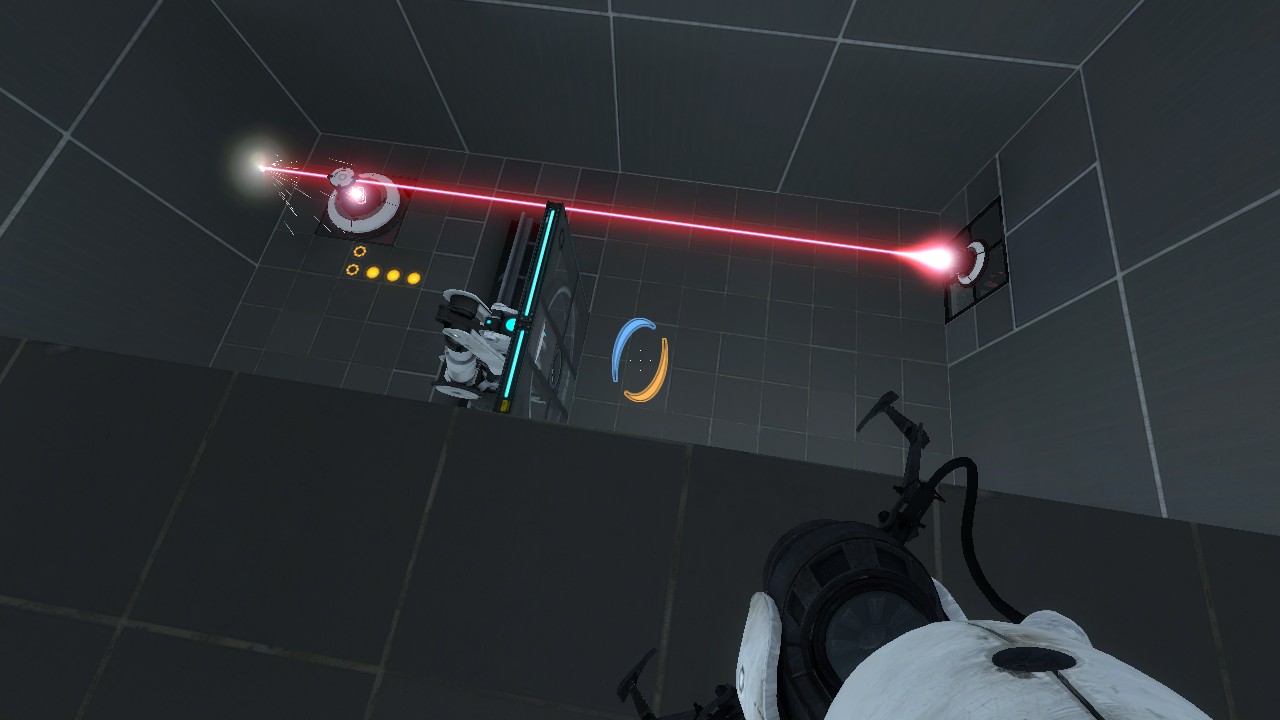}
    \caption{A memory latch in the on state.}
    \label{fig:memory-latch1}
\end{figure}

The last construction is the switch, which we build out of two groups of lasers, moving platforms, and laser relays, as well as a memory latch. The player has the ability to change the state of the memory latch. We interpret the state of the memory latch as the state of the switch. When active, one of the relays in the latch moves a platform out of the way of one of the lasers, activating the corresponding relays and opening the set of doors to which they are connected. Another relay in the latch moves the second moving platform into the path of the second laser, deactivating its corresponding laser relays and the doors they control. Likewise, deactivating the memory latch causes both moving platforms to revert to their original positions, blocking the first laser and letting the second through. We have now successfully constructed doors and switches, so by Metatheorem~\ref{thm:switches} and Theorem~\ref{thm:pspace}, PSPACE-completeness follows.
\end{proof}

Note that in the proof of the preceding theorem, laser catchers could be used in place of laser relays, although the relays have the convenient property that they each need only be connected to a single moving platform. It is also possible that the proof could be adapted to use a single Reflection Cube instead of portals. Additional care would be required with respect to the construction of the door, and it would need to be the case that lasers from multiple directions blocked the avatar. Emancipation Grills or long falls with the moving platforms would simplify this particular door construction. 

The game elements in the following corollary are a superset of those used in Theorem~\ref{thm:cubes-pspace}, so this result follows trivially. However, we prove it by using a construction similar to that in Theorem~\ref{thm:lasers-pspace}, as we feel that the gadgets involved are interesting. We also note that the proof only uses Heavy Duty Super Buttons placed on vertical surfaces, whereas Theorem~\ref{thm:cubes-pspace} relies on their placement on the floor.

\begin{corollary} \label{thm:gravity-pspace}
\textsc{Portal} with gravity beams, cubes, Heavy Duty Super Buttons, and long fall is PSPACE-complete.
\end{corollary}
\begin{proof}
When active, a gravity beam causes objects which fit inside its diameter to be pushed or pulled in line with the gravity beam emitter. Objects in the gravity beam ignore the normal pull of gravity, and thus float along their course. We construct a simple door by placing a gravity beam so that it can carry the player avatar across a pit large enough that the avatar would otherwise be unable to traverse. We hook the gravity beam emitter up to a button allowing it to be turned on and off, unlocking and locking the door.

If we wish to only use buttons placed on vertical surfaces, we are now faced with the problem of making changes to doors persist once the avatar stops holding a cube next to the button. To solve this problem, we construct a memory latch as in Theorem~\ref{thm:lasers-pspace}. If a weighted cube button is placed in the path of a gravity beam, a weighted cube caught in the beam can depress the button as in Figure~\ref{fig:grav-memory-latch1}. A cube on the floor near a gravity beam, \iffull as in Figure~\ref{fig:grav-memory-latch0}\fi~will be picked up by the beam. \iffull Weighted cube buttons can activate and deactivate the same mechanics as laser catchers, including gravity beam emitters. Figures~\ref{fig:grav-memory-latch0} and \ref{fig:grav-memory-latch1} demonstrate a memory latch in the off and on positions, respectively.\fi We also note that gravity beams are blocked by moving platforms, just like lasers. At this point, we have the properties we need from the laser, laser catcher, and moving platform. We also note that the player can pick up and remove cubes from the beam, meaning that portals are not needed. 

\begin{figure}[t]
  \centering
    \includegraphics[width=0.72\textwidth]{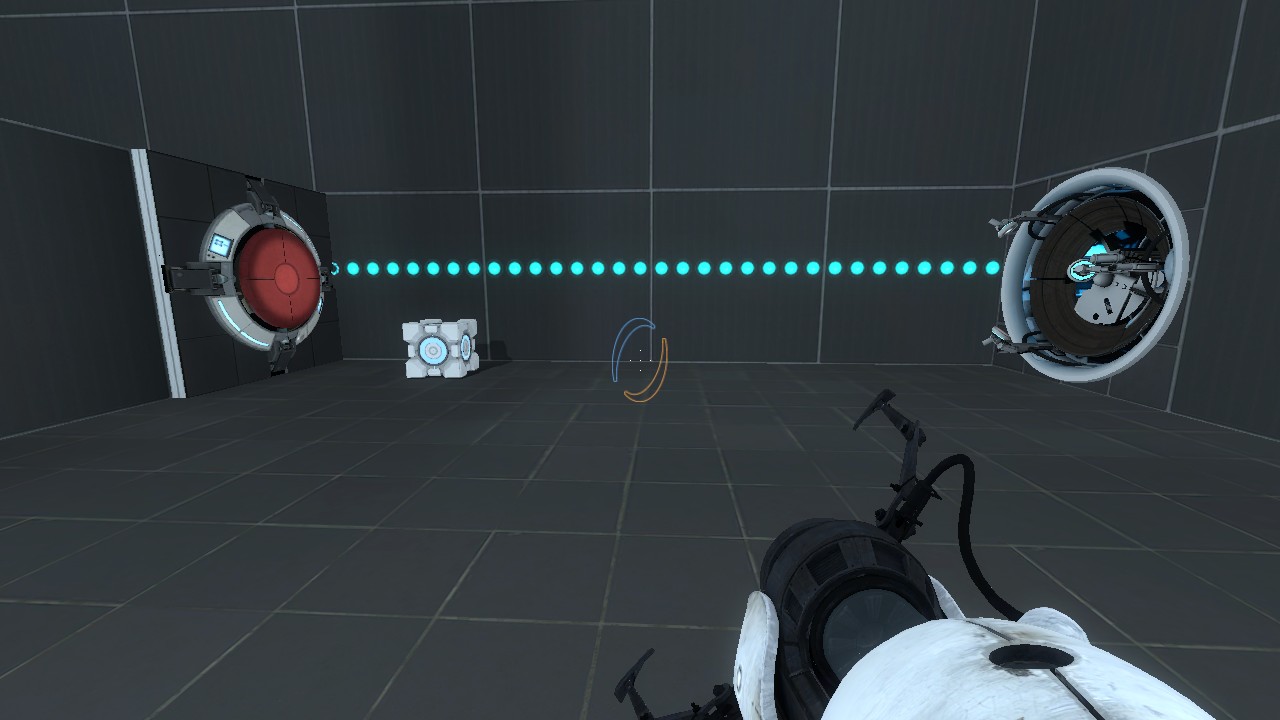}
    \caption{A memory latch in the off state.}
    \label{fig:grav-memory-latch0}
\end{figure}

\begin{figure}[t]
  \centering
    \includegraphics[width=0.72\textwidth]{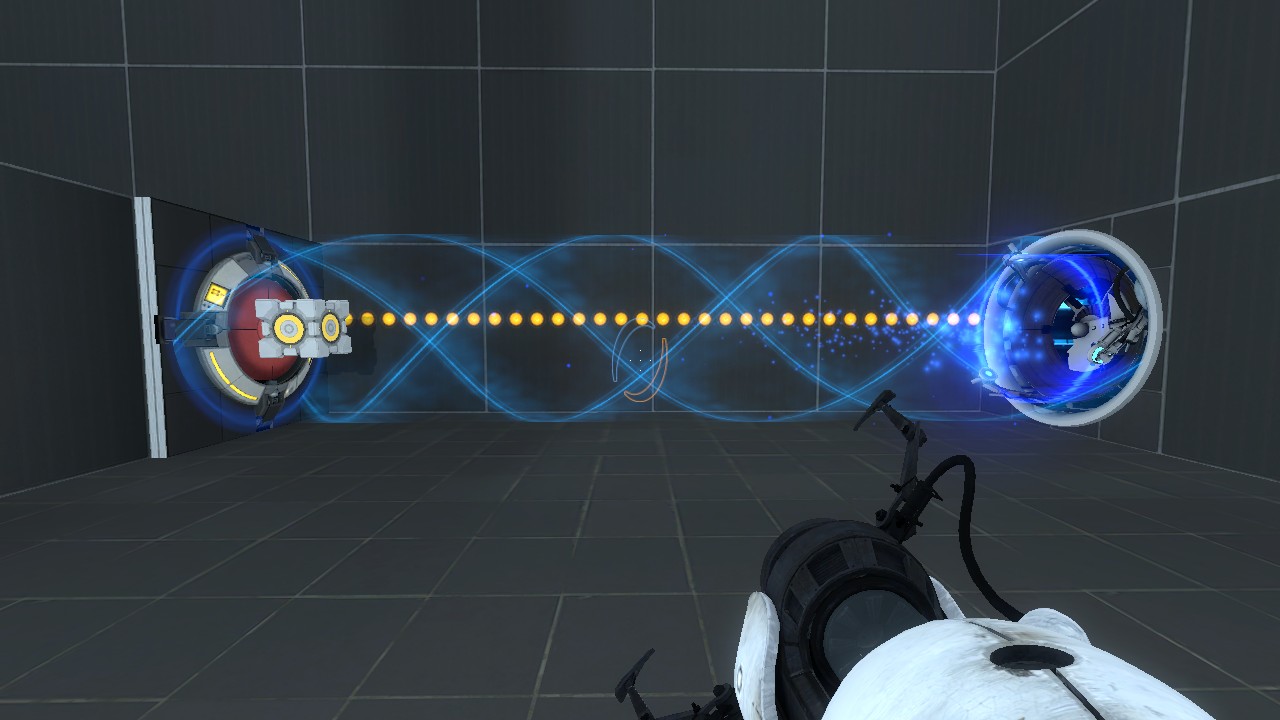}
    \caption{A memory latch in the on state.}
    \label{fig:grav-memory-latch1}
\end{figure}
\end{proof}

\section{Conclusion}
In this paper we proved a number of hardness results about the video game Portal. In Sections \ref{sec:PortalFalling} through \ref{sec:PortalHEP} we have identified several game elements that, when accounted for, give Portal sufficient flexibility so as to encode instances of NP-hard problems. Furthermore, in Section~\ref{sec:pspace} we gave a new metatheorem and use it to prove that certain additional game elements, such as lasers, relays and moving platforms, make the game PSPACE-complete. The unique game mechanics of Portal provided us with a beautiful and unique playground in which to implement the gadgets involved in the hardness proofs. Indeed, our work shows how clause, literal, and variable gadgets inspired by the work of Aloupis et al.~\cite{NintendoFun2014} can be implemented in a 3D video game. 
While our results about Portal itself will be of interest to game and puzzle enthusiasts, what we consider most interesting are the techniques we utilized to obtain them. Adding new, simple gadgets to this collection of abstractions gives us powerful new tools with which to attack future problems. In Section \iffull\ref{subsec:OtherGames}\fi\ifabstract\ref{sec:PortalTurrets}\fi~we identified several other video games that our techniques can be generalized to. We also believe the decomposition of games into individual mechanics will be an important tactic for understanding games of increasing complexity. Metatheorems~\ref{thm:timed-thm} and \ref{thm:switches} are new metatheorems for platform games. We hope that our work is useful as a stepping stone towards more metatheorems of this type. Additionally, we hope the study of motion planning in environments with dynamic topologies leads to new insights in this area.

\subsection{Open Questions}
This work leads to many open questions to pursue in future research. In Portal, we leave many hardness gaps and a number of mechanics unexplored. We are particularly curious about Portal with only portals, and Portal with only cubes. The removal of Emancipation Fields from our proofs would be very satisfying. The other major introduction in Portal 2 that we have not covered is co-op mode. If the players are free to communicate and have perfect information of the map, this feature should not add to the complexity of the game. However, the game seems designed with limited communication in mind and thus an imperfect-information model seems reasonable. Although perfect-information team games tend to reduce down to one- or two-player games, it has been shown that when the players have imperfect information the problem can become significantly harder. In particular, a cooperative game with imperfect information can be 2EXPTIME-complete~\cite{peterson2001lower}, while a team game with imperfect information can be undecidable \cite{demaine2008constraint}. We are not aware of any common or natural games that have used these techniques and think it would be very interesting to have a result such as Bridge or Team Fortress 2 being undecidable.

More than the results themselves, one would hope to use these techniques to show hardness for other problems. Many other games use movable blocks, timed door buttons, and stationary turrets and may have hardness results that immediately follow. Some techniques like encoding numbers in velocities might be transferable.  It would be good to generalize some of these into metatheorems which cover a larger variety of games.

\section*{Acknowledgments}

All raster figures are screenshots from Valve's Portal or Portal 2,
either using Portal 2's Puzzle Maker or
by way of the Portal Unofficial Wiki (\url{http://theportalwiki.com/}).

%% bibliography
\bibliographystyle{alpha}
\bibliography{PortalBib}{}
\end{document}